\documentclass[12pt,leqno]{article}

\input{style}
\numberwithin{equation}{section}
\usepackage{todonotes}

\begin{document}

\title{The stochastic field of aggregate utilities and its saddle
  conjugate.}

\author{P. Bank\\
  Technische Universit{\"a}t Berlin\\
  Institut f{\"u}r Mathematik\\
  Stra{\ss}e des 17. Juni 135, 10623 Berlin, Germany \\
  (bank@math.tu-berlin.de) \and D. Kramkov \thanks{The author also
    holds a part-time position at the University of Oxford. This
    research was supported in part by the Carnegie Mellon-Portugal
    Program and by the Oxford-Man Institute for
    Quantitative Finance at the University of Oxford.} \\
  Carnegie Mellon University,\\
  Department of  Mathematical Sciences,\\
  5000 Forbes Avenue, Pittsburgh, PA, 15213-3890, US \\
  (kramkov@cmu.edu)} \date{\today}

\maketitle
\begin{abstract}
  We describe the sample paths of the stochastic field $F =
  F_t(v,x,q)$ of aggregate utilities parameterized by Pareto weights
  $v$ and total cash amounts $x$ and stocks' quantities $q$ in an
  economy. We also describe the sample paths of the stochastic field
  $G = G_t(u,y,q)$, which is conjugate to $F$ with respect to the
  saddle arguments $(v,x)$, and obtain various conjugacy relations
  between these stochastic fields. The results of this paper play a
  key role in our study in \cite{BankKram:11a}, \cite{BankKram:11b},
  and \cite{BankKram:13} of a continuous-time price impact model.
\end{abstract}

\begin{description}
\item[Keywords:] Envelope theorem, equilibrium, utility indifference
  prices, Pareto allocation, price impact model, risk-aversion,
  risk-tolerance, saddle function, stochastic field.
\item[MSC:] 52A41, 60G60, 91G10, 91G20.
\item[JEL Classification:] G11, G12, G13, C61.
\end{description}

\tableofcontents

\section{Setup and motivation}
\label{sec:setup}

Let $u_m = u_m(x)$, $m=1,\dots,M$, be functions on the real line
$\mathbf{R}$ satisfying
\begin{Assumption}
  \label{as:1}
  Each function $u_m$ is strictly concave, strictly increasing,
  continuously differentiable, and such that
  \begin{equation}
    \label{eq:1}
    \lim_{x\to \infty} u_m(x) = 0. 
  \end{equation} 
\end{Assumption}

The normalization to zero in~\eqref{eq:1} is added only for notational
convenience.  From Assumption~\ref{as:1} we clearly deduce that
\begin{equation}
  \label{eq:2}
  \lim_{x\to-\infty} u_m(x) = -\infty. 
\end{equation}  
Many of our results are derived under the additional condition which,
in particular, implies the boundedness of $u_m$ from above.
\begin{Assumption}
  \label{as:2}
  Each function $u_m$ is twice continuously differentiable and, for
  some constant $c>0$,
  \begin{equation}
    \label{eq:3}
    \frac1c \leq a_m(x)  \set -\frac{u_m''(x)}{u_m'(x)} \leq c,
    \quad x \in \mathbf{R}. 
  \end{equation}
\end{Assumption}

In the model of price impact discussed below, the functions $(u_m)$
and $(a_m)$ describe agents' utilities and absolute
risk-aversions. From Assumptions~\ref{as:1} and~\ref{as:2} we deduce
that
\begin{equation}
  \label{eq:4}
  \frac1c \leq -\frac{u_m'(x)}{u_m(x)} \leq c, \quad x \in \mathbf{R}.  
\end{equation}

Denote by $r=r(v,x)$ the $v$-weighted $\sup$-convolution:
\begin{displaymath}
  r(v,x) \set \max_{x^1+\dots+ x^M = x} \sum_{m=1}^M v^m u_m(x^m),
  \quad (v,x) \in (0,\infty)^M \times \mathbf{R}.
\end{displaymath}
The properties of this function are collected in
Section~\ref{sec:repr-mark-maker}. In particular, for every $v\in
(0,\infty)^M$, the function $r(v,\cdot)$ on $\mathbf{R}$ satisfies
same Assumptions~\ref{as:1} and~\ref{as:2} as each $u_m$. As usual, in
financial economics, we call $r=r(v,x)$ the \emph{aggregate utility
  function}.

Let $\Sigma_0$ and $\psi = (\psi^j)_{j=1,\dots,J}$ be random variables
on a complete filtered probability space $(\Omega, \mathcal{F}_T,
(\mathcal{F}_t)_{0 \leq t \leq T}, \mathbb{P})$ with a finite maturity
$T$. Denote
\begin{displaymath}
  \Sigma(x,q) \set \Sigma_0 + x+\ip{q}{\psi} = \Sigma_0  + x +
  \sum_{j=1}^J q^j \psi^j, \quad (x,q) \in \mathbf{R} 
  \times \mathbf{R}^J,
\end{displaymath}
and assume that
\begin{equation}
  \label{eq:5}
  \mathbb{E}[r(v,\Sigma(x,q))] > -\infty, \quad 
  (v,x,q) \in \mathbf{A}, 
\end{equation}
where 
\begin{equation}
  \label{eq:6}
  \mathbf{A} \set (0,\infty)^M \times \mathbf{R} \times \mathbf{R}^J.
\end{equation}

The main results of the paper, Theorems~\ref{th:9} and \ref{th:10},
describe the sample paths of the stochastic fields $F=F_t(a)$ and $G =
G_t(b)$ given, for $t \in [0,T]$, by
\begin{align*}
  F_t(a) &\set \mathbb{E}[r(v,\Sigma(x,q))|\mathcal{F}_t], \quad
  a=(v,x,q)\in \mathbf{A}, \\
  G_t(b) &\set \sup_{v\in (0,\infty)^M}\inf_{x\in
    \mathbf{R}}[\ip{v}{u} + xy - F_t(v,x,q)], \quad b=(u,y,q) \in
  \mathbf{B},
\end{align*}
where 
\begin{equation}
  \label{eq:7}
  \mathbf{B} \set (-\infty,0)^M \times (0,\infty) \times \mathbf{R}^J. 
\end{equation}
In particular, Theorem~\ref{th:9} shows that these fields have
versions which are differentiable in their spatial arguments and RCLL
with respect to time, while Theorem~\ref{th:10} provides uniform
estimates for their spatial second order derivatives depending only on
the risk-aversion bound of Assumption~\ref{as:2}.  In view of its
construction, we call $F$ the \emph{stochastic field of aggregate
  utilities}.

Of course, the basic regularity properties of $r=r(v,x)$ are
well-known, albeit mostly stated for utility functions on the positive
half line. We refer to \citet{Dana:93, DanaLeVan:96} as well as
\citet{KaratLehocShr:90} and the references therein for pertaining
results, e.g., on the differentiability in $x$ and on the
sensitivities with respect to $v$.  To the best of our knowledge
though, the filtered version of the induced expected utilities as
described by the field $F$ has not been investigated in depth
before. In particular, its structure as an RCLL process in a space of
saddle functions seems to be new. We also could not find a reference
where the saddle conjugate $G$ would be analyzed or made used of.

Our work is motivated by the study of a financial
model with price impact; see the accompanying papers
\cite{BankKram:11a}, \cite{BankKram:11b}, and \cite{BankKram:13}.  In
this model, $M$ market makers quote \emph{utility indifference} prices
for $J$ stocks to a large investor. The market makers' preferences are
specified by the utility functions $(u_m)_{m=1,\dots,M}$ for terminal
wealth. Their initial total endowment is given by $\Sigma_0$. The
stocks pay the terminal dividends $\psi = (\psi^j)_{j=1,\dots,J}$.

If, as a result of trading with a large investor up to time $t$, the
market makers acquire a cash amount $x$ and quantities of stocks $q =
(q^j)$, then their total endowment becomes $\Sigma(x,q)$.  The model
assumes that $\Sigma(x,q)$ is distributed among the market makers in
the form of a Pareto optimal allocation $\pi(a) = (\pi^m(a))$ which,
for $a=(v,x,q)\in \mathbf{A}$, is given by
\begin{equation}
  \label{eq:8}
  v^m u'_m(\pi^m(a)) = \frac{\partial{r}}{\partial
    x}(v,\Sigma(x,q)),\quad m=1,\dots,M.
\end{equation} 
Here the Pareto weight $v\in (0,\infty)^M$ is unique up to a
multiplication on a positive constant.  A trade at time $t$ in $\Delta
q$ stocks results in the new Pareto optimal allocation $\pi(a+\Delta
a) = \pi(x+\Delta x, q+\Delta q, v+\Delta v)$ determined by the
condition of utility indifference:
\begin{displaymath}
  \mathbb{E}[u_m(\pi^m(a))|\mathcal{F}_t] =
  \mathbb{E}[u_m(\pi^m(a+\Delta a))|\mathcal{F}_t], \quad
  m=1,\dots,M. 
\end{displaymath}

More generally, given an $\mathbf{R}^J$-valued demand process $Q =
(Q^j_t)$ for stocks, the model's evolution is described by the
processes $X = (X_t)$, of cash amounts, and $V=(V^m_t)$, of Pareto
weights, solving the equation:
\begin{equation}
  \label{eq:9}
  U^m_t(A_t) = U^m_0(A_0) + \int_0^t U^m(A_s,ds), \quad t\in [0,T],\;
  m=1,\dots,M.     
\end{equation}
Here $A \set (V,X,Q)$ is a predictable process with values in
$\mathbf{A}$, the stochastic field
\begin{displaymath}
  U^m_t(a)  = U^m(a,t) \set
  \mathbb{E}[u_m(\pi^m(a))|\mathcal{F}_t], \quad a\in \mathbf{A}, 
\end{displaymath}
represents the expected utilities for the $m$th market maker, and the
nonlinear stochastic integral of $A$ against this field is defined as
in Section~3.2 of \citet{Kunit:90}.  In particular, if $U^m$ admits
the integral representation
\begin{displaymath}
  U^m_t(a) =  U^m_0(a) + \int_0^t K^m_s(a)dB_s
\end{displaymath}
with respect to a $d$-dimensional Brownian motion $B$ such that the
$d$-dimensional stochastic field $K^m = K^m_s(a)$ has continuous
sample paths with respect to $a$, then
\begin{displaymath}
  \int_0^t U^m(A_s,ds) \set \int_0^t K^m_s(A_s)dB_s. 
\end{displaymath}
If we restrict $V$ to take values in the interior of the simplex:
\begin{equation}
  \label{eq:10}
  \mathbf{S}^M \set \descr{w \in (0,1)^M}{\sum_{m=1}^M w^m = 1},
\end{equation}
then, for given $Q$, the model's evolution is well-defined
if~\eqref{eq:9} has a unique solution $(X,V)$. The mathematical
challenge is then to specify the conditions when this is true.

A natural idea in the study of~\eqref{eq:9} is to substitute $U$ for
$(X,V)$ since then this equation takes a familiar ``explicit'' form.
The related construction of the stochastic fields $X = (X_t(u,q))$ and
$V = (V^m_t(u,q))$ inverse to $U = (U^m_t(x,u,q))$, that is, such that
\begin{displaymath}
  U^m_t(X_t(u,q),V_t(u,q),q) = u^m, \quad m=1,\dots,M,  
\end{displaymath}
is greatly simplified by the observation that $U$ is the $v$-gradient
of $F$:
\begin{displaymath}
  U^m_t(a) = \frac{\partial F_t}{\partial v^m}(a), \quad a = (v,x,q),
  \; m=1,\dots, M. 
\end{displaymath}
The theory of conjugate saddle functions, see Part VII of
\citet{Rock:70}, then allows us to express $X$ and $V$ in terms of the
partial derivatives of $G$:
\begin{align*}
  X_t(u,q) & = \frac{\partial G_t}{\partial y}(u,1,q) =
  G_t(u,1,q), \\
  V^m_t(u,q) & = \frac{\partial G_t}{\partial
    u^m}\frac{1}{\sum_{l=1}^M \frac{\partial G_t}{\partial
      u^l}}(u,1,q).
\end{align*}
Thus, the properties of the stochastic fields $U$, $X$, and $V$ follow
from those of $F$ and $G$, motivating the study of the latter.

We refer the reader to \cite{BankKram:11a} and \cite{BankKram:11b} for
more details on the price impact model and to \cite{BankKram:13} for
convenient sufficient conditions guaranteeing unique solvability
of~\eqref{eq:9}.  These papers use extensively the results of the
current study and explain the economic background as well as related
work.

The paper is organized as follows.  The appropriate conjugate spaces
of saddle functions are defined and studied in
Section~\ref{sec:spac-saddle-funct}. In
Section~\ref{sec:proc-indir-util} these spaces are shown to contain
the sample paths of the stochastic fields $F_t=F_t(a)$ and $G_t =
G_t(b)$. The proofs of the main theorems in
Section~\ref{sec:proc-indir-util} also rely on the properties of the
function $r=r(v,x)$ given in Section~\ref{sec:repr-mark-maker}, on the
``envelope'' theorem for saddle functions stated in
Appendix~\ref{sec:envelope-theorem}, and on the criteria for sample
paths of random and stochastic fields with values in saddle functions
presented in Appendices \ref{sec:unif-integr-saddle}
and~\ref{sec:modif-rand-fields}.

\section{Conjugate spaces of saddle functions}
\label{sec:spac-saddle-funct}

In this section we study the conjugate spaces of saddle functions
which later are shown to support the sample paths of the stochastic
fields $F = F_t(a)$ and $G = G_t(b)$.

Recall some standard notations. For a function $f=f(x,y)$, where $x\in
\mathbf{R}^n$ and $y\in \mathbf{R}^m$, we denote by $\frac{\partial
  f}{\partial x}\set \left(\frac{\partial f}{\partial
    x^1},\dots,\frac{\partial f}{\partial x^n}\right)$ the vector of
partial derivatives with respect to $x$ and by $\nabla f \set
(\frac{\partial f}{\partial x}, \frac{\partial f}{\partial y})$ its
gradient.  For a set $A\subset \mathbf{R}^d$, the notations
$\boundary{A}$ and $\closure{A}$ stand for, respectively, the boundary
and the closure.  For $x,y\in \mathbf{R}^d$ denote $\ip{x}{y} \set
\sum_{i=1}^d x^i y^i$ and $\abs{x} \set \sqrt{\ip{x}{x}}$, the
Euclidean scalar product and the norm.

\subsection{The spaces $\mathbf{F}^1$ and $\mathbf{F}^2$}
\label{sec:spaces-F}

Recall the construction of the parameter set $\mathbf{A}$
from~\eqref{eq:6} and the notation $\mathbf{S}^M$ from~\eqref{eq:10}
for the interior of the simplex in $\mathbf{R}^M$. We shall often
decompose $a\in\mathbf{A}$ as $a=(v,x,q)$, where $v\in (0,\infty)^M$,
$x\in \mathbf{R}$, and $q\in \mathbf{R}^J$.

For a function $\map{f}{\mathbf{A}}{(-\infty,0)}$ define the following
conditions:
\begin{enumerate}[label=(F\arabic{*}), ref=(F\arabic{*})]
\item \label{item:F1} The function $f$ is continuously differentiable
  on $\mathbf{A}$.
\item \label{item:F2} For every $(x,q)\in \mathbf{R} \times
  \mathbf{R}^J$, the function $f(\cdot,x,q)$ is positively
  homogeneous:
  \begin{equation}
    \label{eq:11}
    f(zv,x,q) = zf(v,x,q), \mtext{for all} z>0 \mtext{and} v\in
    (0,\infty)^M, 
  \end{equation}
  and strictly decreasing on $(0,\infty)^M$. Moreover, if $M>1$ then
  $f(\cdot,x,q)$ is strictly convex on $\mathbf{S}^M$ and
  \begin{equation}
    \label{eq:12}
    \lim_{n\to \infty}f(w_n,x,q) = 0, 
  \end{equation}
  for every sequence $(w_n)_{n\geq 1}$ in $\mathbf{S}^M$ converging to
  a boundary point of $\mathbf{S}^M$.
\item \label{item:F3} For every $v\in (0,\infty)^M$, the function
  $f(v,\cdot,\cdot)$ is concave on $\mathbf{R}\times \mathbf{R}^J$.
\item \label{item:F4} For every $(v,q)\in (0,\infty)^M\times
  \mathbf{R}^J$, the function $f(v,\cdot,q)$ is strictly concave and
  strictly increasing on $\mathbf{R}$ and
  \begin{equation}
    \label{eq:13}
    \lim_{x\to\infty} f(v,x,q) = 0.
  \end{equation}
\item \label{item:F5} The function $f$ is twice continuously
  differentiable on $\mathbf{A}$ and, for every $a=(v,x,q)\in
  \mathbf{A}$,
  \begin{displaymath}
    \frac{\partial^2 f}{\partial x^2}(a) < 0, 
  \end{displaymath}
  and the matrix $A(f)(a) = (A^{lm}(f)(a))_{l,m=1,\dots,M}$ given by
  \begin{equation}
    \label{eq:14}
    A^{lm}(f)(a) \set \frac{v^lv^m}{\frac{\partial f}{\partial x}}
    \left(\frac{\partial^2 f}{\partial v^l\partial 
        v^m} -  \frac1{\frac{\partial^2 f}{\partial x^2}} \frac{\partial^2
        f}{\partial v^l\partial 
        x}\frac{\partial^2 f}{\partial v^m\partial
        x}\right)(a),
  \end{equation}
  has full rank.
\end{enumerate}

We now define the families of functions:
\begin{align*}
  \mathbf{F}^1 \set&
  \descr{\map{f}{\mathbf{A}}{(-\infty,0)}}{\text{\ref{item:F1}--\ref{item:F4}
      hold}}, \\
  \mathbf{F}^2 \set& \descr{f\in \mathbf{F}^1}{\text{\ref{item:F5}
      holds}}.
\end{align*}

\begin{Remark}
  \label{rem:1}
  Slightly abusing notations we shall use the same symbols
  $\mathbf{F}^i$, $i=1,2$, for the families of functions $f=f(v,x)$ on
  $(0,\infty)^M\times \mathbf{R}$ whose natural extensions $\widetilde
  f(v,x,q) \set f(v,x)$ to functions on $\mathbf{A}$ belong to
  $\mathbf{F}^i$. Note that, in this case, \ref{item:F3} follows
  trivially from \ref{item:F4}.  A similar convention will also be
  used for other spaces of functions introduced below.
\end{Remark}

\subsection{The spaces $\mathbf{G}^1$ and $\mathbf{G}^2$}
\label{sec:spaces-G}

Recall the construction of the parameter set $\mathbf{B}$
in~\eqref{eq:7}. We shall often decompose $b\in\mathbf{B}$ as
$b=(u,y,q)$, where $u\in (-\infty,0)^M$, $y\in (0,\infty)$, and $q\in
\mathbf{R}^J$.

For a function $\map{g}{\mathbf{B}}{\mathbf{R}}$ define the following
conditions:
\begin{enumerate}[label=(G\arabic{*}), ref=(G\arabic{*})]
\item \label{item:G1} The function $g$ is continuously differentiable
  on $\mathbf{B}$.
\item \label{item:G2} For every $(y,q)\in (0,\infty) \times
  \mathbf{R}^J$, the function $g(\cdot,y,q)$ is strictly increasing
  and strictly convex on $(-\infty,0)^M$. Moreover,
  \begin{enumerate}[label=(\alph{*}), ref=(\alph{*})]
  \item \label{item:1} If $(u_n)_{n\geq 1}$ is a sequence in
    $(-\infty,0)^M$ converging to $0$, then
    \begin{equation}
      \label{eq:15}
      \lim_{n\to \infty} g(u_n,y,q) = \infty.
    \end{equation}
  \item \label{item:2} If $(u_n)_{n\geq 1}$ is a sequence in
    $(-\infty,0)^M$ converging to a boundary point of $(-\infty,0)^M$,
    then
    \begin{equation}
      \label{eq:16}
      \lim_{n\to \infty} \abs{\frac{\partial g}{\partial u}(u_n,y,q)} =
      \infty.
    \end{equation}
  \item \label{item:3} If $M>1$ and $(u_n)_{n\geq 1}$ is a sequence in
    $(-\infty,0)^M$ such that
    \begin{equation}
      \label{eq:17}
      \limsup_{n\to \infty} u_n^m<0 \mtext{for all} m=1,\dots,M
    \end{equation}
    and
    \begin{equation}
      \label{eq:18}
      \lim_{n\to \infty} u_n^{m_0} = -\infty \mtext{for some} m_0 \in
      \braces{1,\dots,M},
    \end{equation}
    then
    \begin{equation}
      \label{eq:19}
      \lim_{n\to\infty} g(u_n,y,q) = -\infty. 
    \end{equation}
  \end{enumerate}
\item \label{item:G3} For every $y \in (0,\infty)$, the function
  $g(\cdot,y,\cdot)$ is convex on $(-\infty,0)^M\times \mathbf{R}^J$.
\item \label{item:G4} For every $(u,q)\in (-\infty,0)^M\times
  \mathbf{R}^J$, the function $g(u,\cdot,q)$ is positively
  homogeneous, that is,
  \begin{equation}
    \label{eq:20}
    g(u,y,q) = yg(u,1,q), \quad y>0. 
  \end{equation}
\item \label{item:G5} The function $g$ is twice continuously
  differentiable on $\mathbf{B}$ and, for every $b=(u,y,q)\in
  \mathbf{B}$, the matrix $B(g)(b) = (B^{lm}(g)(b))_{l,m=1,\dots,M}$
  given by
  \begin{equation}
    \label{eq:21}
    B^{lm}(g)(b) \set \frac{y}{\frac{\partial g}{\partial
        u^l}\frac{\partial g}{\partial u^m}} 
    \frac{\partial^2 g}{\partial u^l\partial 
      u^m}(b) 
  \end{equation}
  has full rank.
\end{enumerate}

We define the families of functions
\begin{align*}
  \mathbf{G}^1 \set&
  \descr{\map{g}{\mathbf{B}}{\mathbf{R}}}{\text{\ref{item:G1}--\ref{item:G4}
      hold}}, \\
  \mathbf{G}^2 \set& \descr{g\in \mathbf{G}^1}{\text{\ref{item:G5}
      holds}}.
\end{align*}

\subsection{Conjugacy relations between $\mathbf{F}^1$ and
  $\mathbf{G}^1$}
\label{step:conj-relat-F1-G1}

The following theorem establishes the key conjugacy relations of this
paper.

\begin{Theorem}
  \label{th:1}
  A function $\map{f}{\mathbf{A}}{(-\infty,0)}$ belongs to
  $\mathbf{F}^1$ if and only if there is $g \in \mathbf{G}^1$ which is
  conjugate to $f$ in the sense that, for every $(u,y,q) \in
  \mathbf{B}$,
  \begin{equation}
    \label{eq:22}
    \begin{split}
      g(u,y,q) &= \sup_{v\in (0,\infty)^M}\inf_{x\in
        \mathbf{R}}[\ip{v}{u} + xy
      - f(v,x,q)]\\
      &= \inf_{x\in \mathbf{R}}\sup_{v\in (0,\infty)^M} [\ip{v}{u} +
      xy - f(v,x,q)],
    \end{split}
  \end{equation}
  and, for every $(v,x,q)\in \mathbf{A}$,
  \begin{equation}
    \label{eq:23}
    \begin{split}
      f(v,x,q) &= \sup_{u\in (-\infty,0)^M}\inf_{y\in
        (0,\infty)}[\ip{v}{u} +
      xy - g(u,y,q)], \\
      &= \inf_{y\in (0,\infty)}\sup_{u\in (-\infty,0)^M} [\ip{v}{u} +
      xy - g(u,y,q)].
    \end{split}
  \end{equation}

  The minimax values in \eqref{eq:22} and \eqref{eq:23} are attained
  at unique saddle points and, for every fixed $q\in \mathbf{R}^J$,
  the following conjugacy relationships between $(v,x)\in
  (0,\infty)^M\times\mathbf{R}$ and $(u,y)\in (-\infty,0)^M\times
  (0,\infty)$ are equivalent:
  \begin{enumerate}
  \item \label{item:4} Given $(u,y)$, the minimax values in
    \eqref{eq:22} are attained at $(v,x)$.
  \item \label{item:5} Given $(v,x)$, the minimax values in
    \eqref{eq:23} are attained at $(u,y)$.
  \item \label{item:6} We have $x = \frac{\partial g}{\partial
      y}(u,y,q) = g(u,1,q)$ and $v = \frac{\partial g}{\partial
      u}(u,y,q)$.
  \item \label{item:7} We have $y = \frac{\partial f}{\partial
      x}(v,x,q)$ and $u = \frac{\partial f}{\partial v}(v,x,q)$.
  \end{enumerate}
  Moreover, in this case, $f(v,x,q) = \ip{u}{v}$, $g(u,y,q) = xy$, and
  \begin{equation}
    \label{eq:24}
    \frac{\partial g}{\partial q}(u,y,q) = - \frac{\partial f}{\partial
      q}(v,x,q).
  \end{equation}
\end{Theorem}

\subsubsection{Proof of Theorem~\ref{th:1}}
\label{sec:proof-theorem-2}

The proof relies on the theory of saddle functions presented in Part
VII of the classical book \cite{Rock:70} by \citeauthor{Rock:70}.  To
simplify notations we omit the dependence on $q$, where it is not
important, and then interpret the classes $\mathbf{F}^1$ and
$\mathbf{G}^1$ in the sense of Remark~\ref{rem:1}.

\begin{Lemma}
  \label{lem:1}
  Let $\map{f=f(v,x)}{(0,\infty)^M\times \mathbf{R}}{(-\infty,0)}$ be
  in $\mathbf{F}^1$. Then there exists a continuously differentiable
  $\map{g = g(u,y)}{(-\infty,0)^M \times (0,\infty)}{\mathbf{R}}$,
  which is conjugate to $f$ in the sense that, for every $u\in
  (-\infty,0)^M$ and $y\in (0,\infty)$,
  \begin{equation}
    \label{eq:25}
    \begin{split}
      g(u,y) &= \sup_{v\in (0,\infty)^M}\inf_{x\in
        \mathbf{R}}[\ip{v}{u} + xy
      - f(v,x)]\\
      &= \inf_{x\in \mathbf{R}}\sup_{v\in (0,\infty)^M} [\ip{v}{u} +
      xy - f(v,x)],
    \end{split}
  \end{equation}
  and, for every $v\in (0,\infty)^M$ and $x\in \mathbf{R}$,
  \begin{equation}
    \label{eq:26}
    \begin{split}
      f(v,x) &= \sup_{u\in (-\infty,0)^M}\inf_{y\in
        (0,\infty)}[\ip{v}{u} +
      xy - g(u,y)], \\
      &= \inf_{y\in (0,\infty)}\sup_{u\in (-\infty,0)^M} [\ip{v}{u} +
      xy - g(u,y)].
    \end{split}
  \end{equation}
  Moreover, the minimax values in~\eqref{eq:25} and~\eqref{eq:26} are
  attained at unique saddle points.
\end{Lemma}

\begin{proof}
  To facilitate the references to Section 37 in \cite{Rock:70} we
  define $f$ on the whole Euclidean space $\mathbf{R}^{M+1}$ by
  setting its values outside of $(0,\infty)^M \times \mathbf{R}$ as
  \begin{equation}
    \label{eq:27}
    f(v,x) = \left\{
      \begin{aligned}
        0, & \quad v \in \boundary{\mathbf{R}^M_+} \\
        \infty, & \quad v\not\in \mathbf{R}^M_+
      \end{aligned}
    \right., \quad x \in \mathbf{R},
  \end{equation}
  where $\mathbf{R}^M_+ \set [0,\infty)^M$.  By \eqref{eq:11} and
  \eqref{eq:12}, after this extension $f$ becomes a \emph{closed}
  saddle function (according to the definition in Section 34 of
  \cite{Rock:70}) with \emph{effective domain}
  \begin{displaymath}
    \dom{f} \set \dom_1{f} \times \dom_2{f} =  \mathbf{R}^M_+ \times
    \mathbf{R}, 
  \end{displaymath}
  where
  \begin{align*}
    \dom_1{f} &\set \descr{v\in \mathbf{R}^M}{f(v,x) < \infty, \quad
      \forall
      x \in \mathbf{R}} = \mathbf{R}^M_+, \\
    \dom_2{f} &\set \descr{x\in \mathbf{R}}{f(v,x) > -\infty, \quad
      \forall v\in \mathbf{R}^M } = \mathbf{R}.
  \end{align*}
  
  Using the extended version of $f$ we introduce the saddle functions
  \begin{align*}
    \underline{g}(u,y) & \set \sup_{v\in \mathbf{R}^M}\inf_{x\in
      \mathbf{R}}[\ip{v}{u} + xy
    - f(v,x)] \\
    & = \sup_{v\in \mathbf{R}^M_+}\inf_{x\in \mathbf{R}}[\ip{v}{u} +
    xy
    - f(v,x)], \\
    \overline{g}(u,y) & \set \inf_{x\in \mathbf{R}}\sup_{v\in
      \mathbf{R}^M}
    [\ip{v}{u} + xy - f(v,x)] \\
    & = \inf_{x\in \mathbf{R}}\sup_{v\in \mathbf{R}^M_+} [\ip{v}{u} +
    xy - f(v,x)]
  \end{align*}
  defined for $u\in \mathbf{R}^M$ and $y\in \mathbf{R}$ and taking
  values in $[-\infty,\infty]$.  By the duality theory for conjugate
  saddle functions, see \cite{Rock:70}, Theorem 37.1 and Corollaries
  37.1.1 and 37.1.2, the functions $\underline{g}$ and $\overline{g}$
  have a common \emph{effective} domain, which we denote $C\times D$,
  and coincide on $(\interior{C}\times D) \cup (C \times
  \interior{D})$, where $\interior{A}$ denotes the interior of a set
  $A$.

  Hence, on $(\interior{C}\times D) \cup (C \times \interior{D})$ we
  can define a finite saddle function $g = g(u,y)$ such that
  \begin{equation}
    \label{eq:28}
    \begin{split}
      g(u,y) &= \sup_{v\in \mathbf{R}^M_+}\inf_{x\in
        \mathbf{R}}[\ip{v}{u} + xy
      - f(v,x)] \\
      &= \inf_{x\in \mathbf{R}}\sup_{v\in \mathbf{R}^M_+} [\ip{v}{u} +
      xy - f(v,x)].
    \end{split}
  \end{equation}
  Moreover, from the same Theorem 37.1 and Corollaries 37.1.1 and
  37.1.2 in \cite{Rock:70} and since \eqref{eq:27} is the unique
  closed extension of $f$ to $\mathbf{R}^{M+1}$ we deduce
  \begin{equation}
    \label{eq:29}
    \begin{split}
      f(v,x) &= \sup_{u\in C}\inf_{y\in \interior{D}}[\ip{v}{u} +
      xy - g(u,y)], \\
      &= \inf_{y\in D}\sup_{u\in \interior{C}} [\ip{v}{u} + xy -
      g(u,y)], \quad (v,x)\in \mathbf{R}^{M+1}.
    \end{split}
  \end{equation}

  Noting that the continuous differentiability of $g$ on
  $(-\infty,0)^M \times (0,\infty)$ is an immediate consequence of the
  existence and the uniqueness of the saddle points for \eqref{eq:25},
  see \cite{Rock:70}, Theorem 35.8 and Corollary 37.5.3, we obtain
  that the result holds if
  \begin{enumerate}
  \item the interiors of the sets $C$ and $D$ are given by
    \begin{align}
      \label{eq:30}
      \interior{C} &= (-\infty,0)^M, \\
      \label{eq:31}
      \interior{D} &= (0,\infty);
    \end{align}
  \item for $(u,y) \in (-\infty,0)^M \times (0,\infty)$, the minimax
    values in \eqref{eq:28} are attained at a unique $(v,x) \in
    (0,\infty)^M \times \mathbf{R}$;
  \item for $(v,x) \in (0,\infty)^M \times \mathbf{R}$, the minimax
    values in \eqref{eq:29} are attained at a unique $(u,y) \in
    (-\infty,0)^M \times (0,\infty)$.
  \end{enumerate}

  For the set $C$ we have
  \begin{align*}
    C & \set \descr{u\in \mathbf{R}^M}{\overline{g}(u,y) < \infty
      \mtext{for all}
      y \in \mathbf{R} } \\
    & = \descr{u\in \mathbf{R}^M}{\sup_{v\in \mathbf{R}^M_+}[\ip{u}{v}
      - f(v,x)] < \infty \mtext{for some}
      x \in \mathbf{R} } \\
    & = \descr{u\in \mathbf{R}^M}{\sup_{w\in \mathbf{S}^M} [\ip{u}{w}
      - f(w,x)] \leq 0 \mtext{for some} x \in \mathbf{R} },
  \end{align*}
  where at the last step we used \eqref{eq:11}.  As $f\leq 0$ on
  $\mathbf{R}^M_+ \times \mathbf{R}$, we have $C\subset
  (-\infty,0]^M$.  On the other hand, by \eqref{eq:12} and
  \eqref{eq:13}, and, since, for every $w\in \mathbf{S}^M$, the
  function $f(w,\cdot)$ is increasing,
  \begin{displaymath}
    \lim_{x\to \infty} \inf_{w\in \mathbf{S}^M} f(w,x) = 0. 
  \end{displaymath}
  It follows that $(-\infty,0)^M \subset C$, proving \eqref{eq:30}.

  For the set $D$ we obtain
  \begin{align*}
    D & \set \descr{y\in \mathbf{R}}{\underline{g}(u,y) > -\infty
      \mtext{for all} u\in \mathbf{R}^M } \\
    & = \descr{y\in \mathbf{R}}{\inf_{x\in \mathbf{R}}[xy - f(v,x)] >
      -\infty \mtext{for some} v\in \mathbf{R}^M }.
  \end{align*}
  From \eqref{eq:13} we deduce that ${D} \subset \mathbf{R}_+$. As
  $f\leq 0$ on $\mathbf{R}^M_+ \times \mathbf{R}$, the point $y=0$
  belongs to $D$. If $y>0$, then \eqref{eq:11} implies the existence
  of $v\in (0,\infty)^M$ such that
  \begin{displaymath}
    y = \frac{\partial f}{\partial x}(v,1),
  \end{displaymath}
  and, therefore, for such $y$ and $v$,
  \begin{displaymath}
    \inf_{x\in \mathbf{R}}[xy - f(v,x)] = y - f(v,1) > -\infty. 
  \end{displaymath}
  Hence, $D= \mathbf{R}_+$, implying \eqref{eq:31}.

  Fix $v\in (0,\infty)^M$ and $x\in \mathbf{R}$. By the properties of
  $f$,
  \begin{displaymath}
    \nabla f (v,x) \in (-\infty,0)^M \times (0,\infty) = \interior{C} \times
    \interior{D}, 
  \end{displaymath}
  implying that $(u,y) \set \nabla f(v,x)$ is the unique saddle point
  of \eqref{eq:29}, see Corollary 37.5.3 in \cite{Rock:70}.

  Fix now $u\in (-\infty,0)^M$ and $y\in (0,\infty)$. As $f$ (viewed
  as a function on $\mathbf{R}^{M+1}$) is a closed saddle function and
  $(u,y)$ belongs to the \emph{interior} of the effective domain of
  $g$, the minimax values in \eqref{eq:28} are attained on a closed
  convex set of saddle points, namely, the subdifferential of $g$
  evaluated at $(u,y)$, see Corollary 37.5.3 in \cite{Rock:70}. To
  complete the proof it remains to be shown that this set is a
  singleton in $(0,\infty)^M\times \mathbf{R}$.

  If $(\widehat{v},\widehat{x})$ is a saddle point of~\eqref{eq:28},
  then $(\widehat{v},\widehat{x}) \in \dom{f} = \mathbf{R}^M_+ \times
  \mathbf{R}$, and
  \begin{align*}
    g(u,y) &= \widehat x y + \ip{\widehat v}{u} - f(\widehat
    v,\widehat x) =\widehat x y + \sup_{v\in \mathbf{R}^M_+}
    [\ip{v}{u} -
    f(v,\widehat x)] \\
    &= \ip{\widehat v}{u} + \inf_{x\in \mathbf{R}} [x y - f(\widehat
    v, x)].
  \end{align*}
  Accounting for the positive homogeneity property \eqref{eq:11} of
  $f(\cdot, \widehat x)$ we deduce that
  \begin{align}
    \label{eq:32}
    \widehat x y  &= g(u,y), \\
    \label{eq:33}
    \ip{\widehat v}{u} - f(\widehat v,\widehat x) &= \sup_{v\in
      \mathbf{R}^M_+} [\ip{v}{u} - f(v,\widehat x)]
    = 0, \\
    \label{eq:34}
    \widehat x y - f(\widehat v,\widehat x) & = \inf_{x\in \mathbf{R}}
    [x y - f(\widehat v, x)].
  \end{align}

  The equality \eqref{eq:32} defines $\widehat x$ uniquely.  To show
  the uniqueness of $\widehat v$ we observe first that $\widehat v \in
  (0,\infty)^M$.  Indeed, otherwise, we would have $f(\widehat v,x) =
  0$, $x\in \mathbf{R}$, and the right side of \eqref{eq:34} would be
  $-\infty$. Hence, $\widehat v$ can be decomposed as a product of
  $\widehat w \in \mathbf{S}^M$ and $\widehat z>0$. By \eqref{eq:11}
  and \eqref{eq:33},
  \begin{displaymath}
    \ip{\widehat w}{u} - f(\widehat w,\widehat x) = \sup_{w\in
      \mathbf{S}^M} [\ip{w}{u} - f(w,\widehat x)]
    = 0.
  \end{displaymath}
  As $f(\cdot,\widehat x)$ is strictly convex on $\mathbf{S}^M$, this
  identity determines $\widehat w$ uniquely. Finally, from
  \eqref{eq:34} and the continuous differentiability of $f(\widehat
  v,\cdot)$ on $\mathbf{R}$ we deduce that
  \begin{displaymath}
    y = \frac{\partial f}{\partial x}(\widehat v,\widehat x) =
    \widehat z \frac{\partial f}{\partial x}(\widehat w,\widehat x), 
  \end{displaymath}
  proving the uniqueness of $\widehat z$.
\end{proof}

\begin{Lemma}
  \label{lem:2}
  Let $f$ and $g$ be as in Lemma~\ref{lem:1}. Then $g$ satisfies the
  positive homogeneity property~\eqref{eq:20} and for $(v,x)\in
  (0,\infty)^M\times\mathbf{R}$ and $(u,y)\in (-\infty,0)^M\times
  (0,\infty)$ the relations below are equivalent:
  \begin{enumerate}
  \item \label{item:8} Given $(u,y)$ the minimax values
    in~\eqref{eq:25} are attained at $(v,x)$.
  \item \label{item:9} Given $(v,x)$ the minimax values
    in~\eqref{eq:26} are attained at $(u,y)$.
  \item \label{item:10} We have $x = \frac{\partial g}{\partial
      y}(u,y) = g(u,1)$ and $v = \frac{\partial g}{\partial u}(u,y)$.
  \item \label{item:11} We have $y = \frac{\partial f}{\partial
      x}(v,x)$ and $u = \frac{\partial f}{\partial v}(v,x)$.
  \end{enumerate}
  Moreover, in this case, $f(v,x) = \ip{u}{v}$ and $g(u,y) = xy$.
\end{Lemma}

\begin{proof}
  First, we observe that \eqref{eq:20} for $g$ follows from the
  corresponding feature \eqref{eq:11} for $f$ and the construction of
  $g$ in \eqref{eq:25}.  The equivalence of
  items~\ref{item:8}--\ref{item:11} follows from the characterization
  of saddle points in terms of the subdifferentials of conjugate
  functions, see Theorem 37.5 and Corollary 37.5.3 in
  \cite{Rock:70}. The last assertion is straightforward.
\end{proof}

\begin{Lemma}
  \label{lem:3}
  Let $f$ and $g$ be as in Lemma~\ref{lem:1}. Then $g\in
  \mathbf{G}^1$.
\end{Lemma}

\begin{proof} The continuous differentiability of $g=g(u,y)$ and the
  positive homogeneity condition \ref{item:G4} have been already
  established in Lemmas~\ref{lem:1} and~\ref{lem:2}, while
  \ref{item:G3} (for $g$ not depending on $q$) follows trivially from
  \ref{item:G2}. Hence, \ref{item:G2} is the only remaining property
  to be proved. In view of \ref{item:G4} it is sufficient to verify it
  only for $y=1$.

  The function $g(\cdot,1)$ is strictly increasing because its
  gradient is strictly positive in view of the characterization of
  saddle points to~\eqref{eq:25} by item~\ref{item:10} of
  Lemma~\ref{lem:2}. To show the strict convexity of $g(\cdot,1)$,
  select $u_1$ and $u_2$, distinct elements of $(-\infty,0)^M$, denote
  by $u_3$ their midpoint, and, for $i=1,2,3$, set $v_i \set
  \frac{\partial g}{\partial u}(u_i,1)$, $x_i \set g(u_i,1)$. Since
  each $v_i \in (0,\infty)^M$, we can represent it as the product $v_i
  = z_i w_i$ of $z_i \in (0,\infty)$ and $w_i \in \mathbf{S}^M$. From
  Lemma~\ref{lem:2} and the positive homogeneity
  property~\eqref{eq:11} we deduce that
  \begin{align}
    \label{eq:35}
    \ip{u_i}{w_i} - f(w_i,x_i) &= \sup_{w\in \mathbf{S}^M}[\ip{u_i}{w}
    -
    f(w,x_i)] = 0, \\
    \label{eq:36}
    (y=) \; 1 &= z_i \frac{\partial f}{\partial x}(w_i,x_i).
  \end{align}
  Since $(u_i)_{i=1,2,3}$ are distinct, so are $(v_i,x_i)_{i=1,2,3}$
  (as $u_i = \frac{\partial f}{\partial v}(v_i,x_i)$ by the
  equivalence of items~\ref{item:10} and~\ref{item:11} in
  Lemma~\ref{lem:2}) and, hence, by \eqref{eq:36}, so are
  $(w_i,x_i)_{i=1,2,3}$. As $f(v,\cdot)$ is strictly concave on
  $\mathbf{R}$, we deduce from \eqref{eq:35} that
  \begin{align*}
    f(w_3,x_3) & = \ip{u_3}{w_3} = \frac12 ((\ip{u_1}{w_3} -
    f(w_3,x_1)) + (\ip{u_2}{w_3} -
    f(w_3,x_2))) \\
    &\quad + \frac12 (f(w_3,x_1) + f(w_3,x_2)) < f(w_3,\frac12(x_1
    +x_2)).
  \end{align*}
  Since $f(w_3,\cdot)$ is strictly increasing on $\mathbf{R}$, we
  deduce that $x_3 < (x_1 + x_2)/2$, implying the strict convexity of
  $g(\cdot,1)$.

  The assertion \eqref{eq:15} of \ref{item:G2}\ref{item:1} follows
  from the monotonicity of $g(\cdot,1)$ and the fact that, by
  Lemmas~\ref{lem:1} and~\ref{lem:2}, for every $x\in \mathbf{R}$ one
  can find $u\in (-\infty,0)^M$ such that $x = g(u,1)$.

  For the proof of \ref{item:G2}\ref{item:2} and
  \ref{item:G2}\ref{item:3} we shall argue by contradiction. Let
  $(u_n)_{n\geq 1}$ be a sequence in $(-\infty,0)^M$ converging to a
  boundary point of $(-\infty,0)^M$.  Denote $x_n \set g(u_n,1)$, $v_n
  \set \frac{\partial g}{\partial u}(u_n,1)$, $n\geq 1$, and, contrary
  to \eqref{eq:16}, assume that the sequence $(v_n)_{n\geq 1}$ is
  bounded.  Then, by the convexity of $g(\cdot,1)$ and the boundedness
  of $(u_n)_{n\geq 1}$, the sequence $(x_n)_{n\geq 1}$ is also
  bounded. Hence, by passing to a subsequence, we can assume that the
  sequences $(v_n)_{n\geq 1}$ and $(x_n)_{n\geq 1}$ converge to
  \emph{finite} limits $\widehat v \in \mathbf{R}^M_+$ and $\widehat x
  \in \mathbf{R}$, respectively.  From Lemma~\ref{lem:2} we deduce
  that
  \begin{equation}
    \label{eq:37}
    u_n = \frac{\partial f}{\partial v}(v_n,x_n), \quad 1 =\frac{\partial
      f}{\partial x}(v_n,x_n), \quad n\geq 1.
  \end{equation}
  If $\widehat v\in (0,\infty)^M$, then
  \begin{displaymath}
    \lim_{n\to \infty} u_n = \lim_{n\to \infty} \frac{\partial f}{\partial
      v}(v_n,x_n) = \frac{\partial f}{\partial v}(\widehat v,\widehat x) \in
    (-\infty,0)^M,  
  \end{displaymath}
  contradicting our choice of $(u_n)_{n\geq 1}$. If, on the other
  hand, $\widehat v\in \boundary{\mathbf{R}^M_+}$, then, by
  \eqref{eq:11} and \eqref{eq:12},
  \begin{displaymath}
    \lim_{n\to \infty} f(v_n,x) = 0, \quad x\in \mathbf{R}. 
  \end{displaymath}
  Since the functions $f(v_n,\cdot)$ are concave, their pointwise
  convergence to $0$ implies the convergence to $0$ of its
  derivatives, uniformly on compact sets in $\mathbf{R}$, see Theorem
  25.7 in \cite{Rock:70}. It follows that
  \begin{displaymath}
    \lim_{n\to \infty} \frac{\partial f}{\partial x}(v_n,x_n) = 0,
  \end{displaymath}
  contradicting the second equality in \eqref{eq:37}. This finishes
  the proof of \ref{item:G2}\ref{item:2}.

  Let now $(u_n)_{n\geq 1}$ be a sequence in $(-\infty,0)^M$
  satisfying the conditions \eqref{eq:17} and \eqref{eq:18} of
  \ref{item:G2}\ref{item:3}. Denote $x_n \set g(u_n,1)$ and, contrary
  to \eqref{eq:19}, assume that
  \begin{displaymath}
    \limsup_{n\to \infty} g(u_n,1) = \limsup_{n\to \infty} x_n > -\infty.
  \end{displaymath}
  As $g(\cdot,1)$ is an increasing function on $(-\infty,0)^M$,
  \eqref{eq:17} implies the boundedness of the sequence $(x_n)_{n\geq
    1}$ from above. Hence, by passing, if necessary, to a subsequence,
  we can assume that it converges to $\widehat x\in \mathbf{R}$.
  Define a sequence $(w_n)_{n\geq 1}$ in $\mathbf{S}^M$ by
  \begin{displaymath}
    w_n = \frac{\partial g}{\partial u}(u_n,y_n) = y_n \frac{\partial
      g}{\partial u}(u_n,1),
  \end{displaymath}
  for appropriate normalizing constants $y_n$, $n\geq 1$.  By passing
  to a subsequence, we can assume that $(w_n)_{n\geq 1}$ converges to
  $\widehat w$ in the simplex $\closure{\mathbf{S}^M}$. From
  Lemma~\ref{lem:2} we deduce that
  \begin{displaymath}
    u_n = \frac{\partial f}{\partial v}(w_n,x_n), \quad n\geq 1.
  \end{displaymath}
  If $\widehat w\in \mathbf{S}^M$, then
  \begin{displaymath}
    \lim_{n\to \infty} u_n = \lim_{n\to \infty} \frac{\partial f}{\partial
      v}(w_n,x_n) = \frac{\partial f}{\partial v}(\widehat w,\widehat x) \in
    (-\infty,0)^M, 
  \end{displaymath}
  contradicting \eqref{eq:18}. If $\widehat w\in
  \boundary{\mathbf{S}^M}$, then, by \eqref{eq:12}, the sequence of
  concave functions $f(w_n,\cdot)$, $n\geq 1$, on $\mathbf{R}$
  converges to $0$ pointwise and, therefore, also uniformly on compact
  sets. Accounting for Lemma~\ref{lem:2} we deduce
  \begin{displaymath}
    \lim_{n\to \infty} \ip{u_n}{w_n} = \lim_{n\to \infty} f(w_n,x_n) = 0, 
  \end{displaymath}
  contradicting \eqref{eq:17}. This finishes the proof of
  \ref{item:G2}\ref{item:3} and, with it, the proof of the lemma.
\end{proof}

\begin{Lemma}
  \label{lem:4}
  Let $\map{g = g(u,y)}{(-\infty,0)^M \times (0,\infty)}{\mathbf{R}}$
  be in $\mathbf{G}^1$. Then there is a continuously differentiable
  function $\map{f=f(v,x)}{(0,\infty)^M\times
    \mathbf{R}}{(-\infty,0)}$ such that the minimax relations
  \eqref{eq:25} and \eqref{eq:26} hold and have unique saddle points.
\end{Lemma}

\begin{proof}
  We follow similar arguments as in the proof of Lemma~\ref{lem:1}.
  To use the results of Section 37 in \cite{Rock:70} we need to define
  the values of $g=g(u,y)$ at the boundary of the original domain by
  an appropriate closure operation. For $u\in (-\infty,0)^M$ we set,
  by continuity, $g(u,0)\set 0$. Then for $y\geq 0$ we define, by
  lower semi-continuity,
  \begin{equation}
    \label{eq:38}
    g(u,y) \set \lim_{\epsilon\to 0} \inf_{z\in B(u,\epsilon)} g(z,y),
    \quad u\in \boundary{(-\infty,0]^M}, 
  \end{equation}
  where
  \begin{displaymath}
    B(u,\epsilon) \set \descr{z\in (-\infty,0)^M}{|u-z|\leq \epsilon }.
  \end{displaymath}
  Note that in \eqref{eq:38} the value of the limit may be infinite.

  As in the proof of Lemma~\ref{lem:1} we deduce the existence of a
  saddle function $f=f(v,x)$ defined on $(C\times \interior{D}) \cup
  (D \times \interior{C})$, where
  \begin{align*}
    C & \set \descr{v\in \mathbf{R}^M}{f(v,x) < \infty
      \mtext{for all} x \in \mathbf{R} } \\
    &= \descr{v\in \mathbf{R}^M}{\sup_{u\in (-\infty,0]^M}[\ip{v}{u} -
      g(u,y)] < \infty \mtext{for some}
      y \in \mathbf{R}_+ }, \\
    D & \set \descr{x\in \mathbf{R}}{f(v,x) > - \infty
      \mtext{for all} v \in \mathbf{R}^M } \\
    &= \descr{x\in \mathbf{R}}{\inf_{y\in \mathbf{R}_+}[xy - g(u,y)] >
      -\infty \mtext{for some} u\in (-\infty,0]^M },
  \end{align*}
  such that, for every $(v,x) \in (C\times \interior{D}) \cup (D
  \times \interior{C})$,
  \begin{equation}
    \label{eq:39}
    \begin{split}
      f(v,x) &= \sup_{u\in (-\infty,0]^M}\inf_{y\in
        \mathbf{R}_+}[\ip{v}{u} +
      xy - g(u,y)], \\
      &= \inf_{y\in \mathbf{R}_+}\sup_{u\in (-\infty,0]^M} [\ip{v}{u}
      + xy - g(u,y)],
    \end{split}
  \end{equation}
  and, for every $(u,y) \in (-\infty,0)^M\times (0,\infty)$,
  \begin{equation}
    \label{eq:40}
    \begin{split}
      g(u,y) &= \sup_{v\in C}\inf_{x\in \interior{D}}[\ip{v}{u} + xy
      - f(v,x)]\\
      &= \inf_{x\in D}\sup_{v\in \interior{C}} [\ip{v}{u} + xy -
      f(v,x)].
    \end{split}
  \end{equation}
 
  As $g(u,y) = yg(u,1)$, and, by \ref{item:G2}\ref{item:3}, for every
  $x\in\mathbf{R}$ there is $u\in (-\infty,0)^M$ such that $x\geq
  g(u,1)$, we have
  \begin{displaymath}
    D = \mathbf{R}. 
  \end{displaymath}

  Choosing $y=0$ in the second description of $C$ above, we obtain
  \begin{displaymath}
    \sup_{u\in (-\infty,0]^M}[\ip{u}{v} - g(u,0)] = \sup_{u\in
      (-\infty,0]^M}[\ip{u}{v}] < \infty \mtext{iff} v\in
    \mathbf{R}^M_+. 
  \end{displaymath}
  If $v\not\in \mathbf{R}^M_+$, then there is $u_0\in (-\infty,0)^M$
  such that $\ip{u_0}{v} > 0$. By \ref{item:G2}\ref{item:3}, for every
  $y>0$,
  \begin{displaymath}
    \lim_{n\to \infty} g(nu_0,y) = -\infty, 
  \end{displaymath}
  and, therefore,
  \begin{displaymath}
    \sup_{u\in (-\infty,0]^M}[\ip{u}{v} - g(u,y)] \geq \limsup_{n\to\infty}
    [\ip{nu_0}{v} - g(nu_0,y)] = \infty. 
  \end{displaymath}
  It follows that
  \begin{displaymath}
    C = \mathbf{R}^M_+.
  \end{displaymath}

  For $u\in (-\infty,0)^M$ and $y>0$ we have $\nabla g(u,y) \in
  (0,\infty)^M \times \mathbf{R}$, implying that $(v,x) \set \nabla
  g(u,y)$ is the unique saddle point of \eqref{eq:40}. In particular,
  we deduce that the minimax identities \eqref{eq:25} and
  \eqref{eq:40} have the same unique saddle points.

  Let now $v\in (0,\infty)^M$ and $x\in \mathbf{R}$. As $(v,x)$
  belongs to the \emph{interior} of the effective domain of $f$, the
  minimax values in \eqref{eq:39} are attained on a closed convex set
  of saddle points belonging to the subdifferential of $f$ evaluated
  at $(v,x)$, see Corollary 37.5.3 in \cite{Rock:70}. We are going to
  show that this set is a singleton in $(-\infty,0)^M\times
  (0,\infty)$.

  Let $(\widehat{u},\widehat{y})$ be a saddle point. Then
  \begin{displaymath}
    (\widehat{u},\widehat{y}) \in \dom{g} \subset (-\infty,0]^M \times
    \mathbf{R}_+, 
  \end{displaymath}
  and
  \begin{align*}
    f(v,x) &= x\widehat y + \ip{\widehat u}{v} - g(\widehat u,\widehat
    y) = x\widehat y + \sup_{u\in (-\infty,0]^M} [\ip{u}{v} -
    g(u,\widehat y)] \\
    &= \ip{\widehat u}{v} + \inf_{y\in \mathbf{R}_+} [x y - g(\widehat
    u, y)].
  \end{align*}
  As $g(u,y) = y g(u,1)$, we deduce that $\widehat y>0$, $\widehat
  u\not=0$, and
  \begin{align}
    \label{eq:41}
    f(v,x) &= \ip{\widehat u}{v}, \\
    \label{eq:42}
    x &= g(\widehat u, 1), \\
    \label{eq:43}
    \ip{\widehat u}{v} - x\widehat y & = \sup_{u\in (-\infty,0]^M }
    [\ip{u}{v} - \widehat y g(u,1)].
  \end{align}

  The attainability of the upper bound in \eqref{eq:43} at $\widehat
  u$ implies that the subdifferential $\partial g(\widehat u,1)$ is
  well-defined and $v\in \widehat y\partial g(\widehat u,1)$. From
  \ref{item:G2}\ref{item:2} we deduce that $g(\cdot,1)$ is not
  subdifferentiable on the boundary of $(-\infty,0]^M$ and, therefore,
  $\widehat u\in (-\infty,0)^M$.  As $g(\cdot,1)$ is strictly convex
  on $(-\infty,0)^M$, \eqref{eq:42} defines $\widehat u$ uniquely,
  and, as $g(\cdot,1)$ is differentiable on $(-\infty,0)^M$, $\widehat
  y$ is uniquely determined by the equality
  \begin{displaymath}
    v = \widehat y \frac{\partial g}{\partial u}(\widehat u,1).  
  \end{displaymath}

  The uniqueness of the saddle points $(\widehat u,\widehat y)$
  implies the continuous differentiability of $f$ on $(0,\infty)^M
  \times \mathbf{R}$. Finally, from \eqref{eq:41} we deduce that $f <
  0$ on $(0,\infty)^M \times \mathbf{R}$.
\end{proof}

\begin{Lemma}
  \label{lem:5}
  Let $g$ and $f$ be as in Lemma~\ref{lem:4}. Then $f$ satisfies the
  positive homogeneity condition \eqref{eq:11} and, for $(v,x)\in
  (0,\infty)^M\times\mathbf{R}$ and $(u,y)\in (-\infty,0)^M\times
  (0,\infty)$, the assertions of Lemma~\ref{lem:2} hold.
\end{Lemma}

\begin{proof}
  The positive homogeneity property \eqref{eq:11} for $f$ is a
  consequence of the corresponding feature \eqref{eq:20} for $g$. The
  remaining assertions follow by the same arguments as in the proof of
  Lemma~\ref{lem:2}.
\end{proof}

\begin{Lemma}
  \label{lem:6}
  Let $g$ and $f$ be as in Lemma~\ref{lem:4}. Then $f$ satisfies
  \ref{item:F2}.
\end{Lemma}

\begin{proof}
  Fix $x\in \mathbf{R}$. The positive homogeneity with respect to $v$
  was already established in Lemma \ref{lem:5}. By item~\ref{item:11}
  of Lemma~\ref{lem:2}, $\frac{\partial f}{\partial v} < 0$, implying
  that the function $f(\cdot,x)$ is strictly decreasing.

  Let $(w_i)_{i=1,2}$ be distinct points in $\mathbf{S}^M$, $w_3$ be
  their midpoint, and, for $i=1,2,3$, denote $u_i \set \frac{\partial
    f}{\partial v}(w_i,x)$ and $y_i \set \frac{\partial f}{\partial
    x}(w_i,x)$.  By the characterizations of saddle points in
  Lemma~\ref{lem:2}, for $i=1,2,3$, we have
  \begin{displaymath}
    f(w_i,x) = \ip{u_i}{w_i}, \quad x = g(u_i,1), \mtext{and} w_i = y_i
    \frac{\partial g}{\partial u}(u_i,1). 
  \end{displaymath}
  From the last equality we deduce that the points $(u_i)_{i=1,2,3}$
  are distinct. The uniqueness of saddle points for \eqref{eq:26} then
  implies $\ip{u_3}{w_i} < \ip{u_i}{w_i}$, for $i=1,2$, and,
  therefore,
  \begin{align*}
    f(w_3,x) &= \ip{u_3}{w_3} = \frac12(\ip{u_3}{w_1} + \ip{u_3}{w_2}) \\
    &< \frac12(\ip{u_1}{w_1} + \ip{u_2}{w_2}) = \frac12(f(w_1,x) +
    f(w_2,x)),
  \end{align*}
  proving the strict convexity of $f(\cdot,x)$ on $\mathbf{S}^M$.

  Let now $(w_n)_{n\geq 1}$ be a sequence in $\mathbf{S}^M$ converging
  to $w\in \boundary{\mathbf{S}^M}$. From \ref{item:G2}\ref{item:3}
  for every $\epsilon >0$ we deduce the existence of $u(\epsilon)\in
  (-\infty,0)^M$ such that $g(u(\epsilon),1)\leq x$ and
  \begin{displaymath}
    -\epsilon \leq \ip{u(\epsilon)}{w} = \lim_{n\to\infty}
    \ip{u(\epsilon)}{w_n}. 
  \end{displaymath}
  From the construction of $f$ in \eqref{eq:26} we deduce $f(v,x) \geq
  \ip{u(\epsilon)}{v}$ for every $v\in (0,\infty)^M$. It follows that
  \begin{displaymath}
    \liminf_{n\to\infty} f(w_n,x) \geq \lim_{n\to\infty}
    \ip{u(\epsilon)}{w_n} \geq -\epsilon, 
  \end{displaymath}
  proving \eqref{eq:12}.
\end{proof}

\begin{Lemma}
  \label{lem:7}
  Let $g$ and $f$ be as in Lemma~\ref{lem:4}. Then $f$ satisfies
  \ref{item:F4}.
\end{Lemma}

\begin{proof}
  Fix $v\in (0,\infty)^M$.  As, by item~\ref{item:11} of
  Lemma~\ref{lem:2}, $\frac{\partial f}{\partial x} > 0$, the function
  $f(v,\cdot)$ is strictly increasing.

  Let $(x_i)_{i=1,2}$ be distinct elements of $\mathbf{R}$, $x_3 \set
  \frac12(x_1+x_2)$, and $u_i \set \frac{\partial f}{\partial
    v}(v,x_i)$, $i=1,2,3$.  From Lemma~\ref{lem:5} we deduce
  \begin{displaymath}
    g(u_i,1) = x_i, \quad f(v,x_i) = \ip{u_i}{v}, \quad i=1,2,3. 
  \end{displaymath}
  It follows that $(u_i)_{i=1,2,3}$ are distinct, and, hence, by the
  strict convexity of $g(\cdot,1)$,
  \begin{displaymath}
    g (\frac12(u_1+u_2),1) < \frac12 (g(u_1,1) + g(u_2,1)) = \frac12 (x_1 +
    x_2) = x_3. 
  \end{displaymath}
  From the uniqueness of saddle points in \eqref{eq:26} we deduce that
  if $g(u,1) < x$ then $f(v,x) > \ip{u}{v}$. It follows that
  \begin{displaymath}
    f(v,x_3) > \ip{\frac12(u_1+u_2)}{v} = \frac12 (f(v,x_1) + f(v,x_2)), 
  \end{displaymath}
  proving the strict concavity of $f(v,\cdot)$.

  For every $\epsilon > 0$ we can clearly find $u(\epsilon)\in
  (-\infty,0)^M$ such that $\ip{u(\epsilon)}{v} \geq
  -\epsilon$. Denoting $x(\epsilon) \set g(u(\epsilon),1)$ we deduce
  \begin{displaymath}
    \lim_{x\to \infty} f(v,x) > f(v,x(\epsilon)) \geq \ip{u(\epsilon)}{v}
    \geq -\epsilon, 
  \end{displaymath}
  proving \eqref{eq:13}.
\end{proof}

After these preparations we are ready to complete the proof of
Theorem~\ref{th:1}. From this moment, the functions $f$ and $g$ will
depend on the ``auxiliary'' variable $q\in \mathbf{R}^J$.

\begin{proof}[Proof of Theorem~\ref{th:1}.] If $f=f(v,x,q) \in
  \mathbf{F}^1$, then, by Lemmas~\ref{lem:1} and~\ref{lem:3}, the
  function
  \begin{displaymath}
    g(u,y,q) \set  \sup_{v\in (0,\infty)^M}\inf_{x\in
      \mathbf{R}}[\ip{v}{u} + xy 
    - f(v,x,q)]
  \end{displaymath}
  satisfies \ref{item:G2} and \ref{item:G4} and is differentiable with
  respect to $u$ and $y$. Moreover, the concavity of
  $f(v,\cdot,\cdot)$ implies the convexity of
  $g(\cdot,y,\cdot)$. Conversely, if $g=g(u,y,q)\in \mathbf{G}^1$,
  then, by Lemmas~\ref{lem:4}--\ref{lem:7}, the function
  \begin{displaymath}
    f(v,x,q) \set \sup_{u\in (-\infty,0)^M}\inf_{y\in (0,\infty)}
    [\ip{u}{v} + xy - g(u,y,q)] 
  \end{displaymath}
  satisfies \ref{item:F2} and \ref{item:F4} and is differentiable with
  respect to $v$ and $x$. Moreover, as $g$ is convex with respect to
  $(u,q)$, $f$ is concave with respect to $(x,q)$.

  The rest of the proof, namely, the equivalence of the
  differentiability of $f$ and $g$ with respect to $q$ and the
  relation \eqref{eq:24}, follows from the \emph{envelope theorem} for
  saddle functions, Theorem~\ref{th:11}, given in
  Appendix~\ref{sec:envelope-theorem}. Finally, we recall that for
  saddle functions the existence of derivatives implies the continuity
  of derivatives, see Theorem 35.8 and Corollary 35.7.1 in
  \cite{Rock:70}.
\end{proof}

\subsection{Conjugacy relations between $\mathbf{F}^2$ and
  $\mathbf{G}^2$}
\label{sec:conj-relat-betw-F2-G2}

For $f\in \mathbf{F}^2$ and $g\in \mathbf{G}^2$, in addition to the
matrices $A(f)$ and $B(g)$ given by \eqref{eq:14} and \eqref{eq:21},
define the following matrices of second derivatives: for $m=1,\dots,M$
and $i,j=1,\dots,J$,
\begin{align}
  \label{eq:44}
  C^{mj}(f)(v,x,q) &\set \frac{v^m}{\frac{\partial f}{\partial x}}
  \left(\frac{\partial^2 f}{\partial v^m \partial q^j} -
    \frac1{\frac{\partial^2 f}{\partial x^2}} \frac{\partial^2
      f}{\partial v^m\partial x} \frac{\partial^2 f}{\partial
      x \partial
      q^j}\right)(v,x,q), \\
  \label{eq:45}
  D^{ij}(f)(v,x,q) &\set \frac{1}{\frac{\partial f}{\partial x}}
  \left(-\frac{\partial^2 f}{\partial q^i \partial q^j} +
    \frac1{\frac{\partial^2 f}{\partial x^2}} \frac{\partial^2
      f}{\partial x\partial q^i} \frac{\partial^2 f}{\partial
      x\partial q^j}\right)(v,x,q),
\end{align}
and
\begin{align}
  \label{eq:46}
  E^{mj}(g)(u,y,q) &\set \frac{1}{\frac{\partial g}{\partial
      u^m}}\frac{\partial^2 g}{\partial u^m \partial q^j}(u,y,q) =
  \frac{1}{\frac{\partial g}{\partial u^m}}\frac{\partial^2
    g}{\partial u^m \partial q^j}(u,1,q), \\
  \label{eq:47}
  H^{ij}(g)(u,y,q) &\set \frac1y \frac{\partial^2 g}{\partial
    q^i \partial q^j} (u,y,q) = \frac{\partial^2 g}{\partial
    q^i \partial q^j}(u,1,q),
\end{align}
where in~\eqref{eq:46} and~\eqref{eq:47} we used the positive
homogeneity~\eqref{eq:20} of $g$ with respect to $y$.

We use standard notations of linear algebra: for a square matrix $A$
of full rank, $A^{-1}$ denotes its inverse, and, for a matrix $B$,
$B^T$ stands for its transpose.

\begin{Theorem}
  \label{th:2}
  A function $\map{f}{\mathbf{A}}{(-\infty,0)}$ belongs to
  $\mathbf{F}^2$ if and only if it is conjugate to a function $g\in
  \mathbf{G}^2$ in the sense that \eqref{eq:22} and \eqref{eq:23}
  hold.

  Moreover, if, for $q\in \mathbf{R}^J$, the vectors $a=(v,x,q)\in
  \mathbf{A}$ and $b=(u,y,q)\in \mathbf{B}$ are conjugate in the sense
  of the equivalent conditions of items~\ref{item:4}--\ref{item:7} of
  Theorem~\ref{th:1}, then the matrices of the second derivatives for
  $f$, $A(f)$, $C(f)$, and $D(f)$, defined in \eqref{eq:14},
  \eqref{eq:44}, and \eqref{eq:45}, and the matrices of the second
  derivatives for $g$, $B(g)$, $E(g)$, and $H(g)$, defined in
  \eqref{eq:21}, \eqref{eq:46}, and \eqref{eq:47}, are related by
  \begin{align}
    \label{eq:48}
    B(g)(b) &= (A(f)(a))^{-1}, \\
    \label{eq:49}
    E(g)(b) &=  - (A(f)(a))^{-1} C(f)(a), \\
    \label{eq:50}
    H(g)(b) &= (C(f)(a))^T (A(f)(a))^{-1} C(f)(a) + D(f)(a).
  \end{align}
\end{Theorem}

\begin{Remark}
  \label{rem:2}
  Our choice of the specific form for the matrices $A(f)(a)$,
  $C(f)(a)$, and $D(f)(a)$ and $B(g)(b)$, $E(g)(b)$, and $H(g)(b)$ was
  partially motivated by the fact that they are invariant under the
  transformations $(v,x,q)\rightarrow (zv,x,q)$ and
  $(u,y,q)\rightarrow (u,zy,q)$, $z>0$, which are natural in light of
  the positive homogeneity conditions \eqref{eq:11} and \eqref{eq:20}.
\end{Remark}

\subsubsection{Proof of Theorem~\ref{th:2}}
\label{sec:proof-theor-refth:2}

As in the proof of Theorem~\ref{th:1} we begin with several lemmas,
where we omit the dependence on $q$.

\begin{Lemma}
  \label{lem:8}
  Let $f=f(v,x)$ and $g=g(u,y)$ be as in Lemma~\ref{lem:1}. Then the
  following assertions are equivalent:
  \begin{enumerate}
  \item $f$ is twice continuously differentiable and for all $v\in
    (0,\infty)^M$ and $x\in \mathbf{R}$ its Hessian matrix $K(v,x) =
    (K^{kl}(v,x))_{k,l=1,\dots,M+1}$ has full rank.
  \item $g$ is twice continuously differentiable and for all $u\in
    (-\infty,0)^M$ and $y\in (0,\infty)$ its Hessian matrix $L(u,y) =
    (L^{kl}(u,y))_{k,l=1,\dots,M+1}$ has full rank.
  \end{enumerate}
  Moreover, if $(v,x)\in (0,\infty)^M\times\mathbf{R}$ and $(u,y)\in
  (-\infty,0)^M\times (0,\infty)$ are conjugate saddle points in the
  sense of Lemma~\ref{lem:2}, then $L(u,y)$ is inverse to $K(v,x)$.
\end{Lemma}

\begin{proof} The asserted equivalence is a well-known fact in the
  theory of saddle functions and is a direct consequence of the
  characterization of the gradients of the conjugate functions $f$ and
  $g$ given in Lemma~\ref{lem:2} and the Implicit Function Theorem.
\end{proof}

In the following statement we shall make the relationship between the
Hessian matrices of $f$ and $g$ more explicit by taking into account
the positive homogeneity property~\eqref{eq:20} of $g$.

\begin{Lemma}
  \label{lem:9}
  Let $f$ and $g$ be as in Lemma~\ref{lem:1}. Then the following
  assertions are equivalent:
  \begin{enumerate}
  \item \label{item:12} $f$ is twice continuously differentiable and
    for all $v\in (0,\infty)^M$ and $x\in \mathbf{R}$
    \begin{equation}
      \label{eq:51}
      \frac{\partial^2 f}{\partial x^2}(v,x) < 0 
    \end{equation}
    and the Hessian matrix $K(v,x)$ of $f$ has full rank.
  \item \label{item:13} $f$ is twice continuously differentiable and
    for all $v\in (0,\infty)^M$ and $x\in \mathbf{R}$ the inequality
    \eqref{eq:51} holds and the $M\times M$ matrix $\widetilde A(v,x)$
    with entries
    \begin{equation}
      \label{eq:52}
      \widetilde{A}^{kl}(v,x) \set 
      \left(\frac{\partial^2 f}{\partial v^k\partial
          v^l} -  \frac1{\frac{\partial^2 f}{\partial x^2}} \frac{\partial^2
          f}{\partial v^k\partial 
          x}\frac{\partial^2 f}{\partial v^l\partial
          x}\right)(v,x),
    \end{equation} 
    has full rank.
  \item \label{item:14} $g$ is twice continuously differentiable and
    for all $u\in (-\infty,0)^M$ and $y\in (0,\infty)$ the $M\times M$
    matrix $\widetilde B(u,y)$ with entries
    \begin{equation}
      \label{eq:53}
      \widetilde{B}^{kl}(u,y) \set \frac{\partial^2 g}{\partial
        u^k\partial u^l}(u,y),
    \end{equation}
    has full rank.
  \end{enumerate}
  Moreover, if $(v,x)\in (0,\infty)^M\times\mathbf{R}$ and $(u,y)\in
  (-\infty,0)^M\times (0,\infty)$ are conjugate saddle points in the
  sense of Lemma~\ref{lem:2}, then $\widetilde{B}(u,y)$ is the inverse
  of $\widetilde{A}(v,x)$.
\end{Lemma}

\begin{proof}
  \ref{item:12} $\Longleftrightarrow$ \ref{item:13}.
  From~\eqref{eq:51} and the construction of the matrix
  $\widetilde{A}(v,x)$ in~\eqref{eq:52} we deduce that for $a\in
  \mathbf{R}^M$ and $b\in \mathbf{R}$ the equation
  \begin{displaymath}
    K(v,x)
    \begin{pmatrix}
      a \\b
    \end{pmatrix}
    = 0
  \end{displaymath}
  is equivalent to
  \begin{displaymath}
    b = - \frac1{\frac{\partial^2 f}{\partial x^2}(v,x)} \sum_{m=1}^M
    \frac{\partial^2 f}{\partial v^m \partial x}(v,x) a^m 
  \end{displaymath}
  and
  \begin{displaymath}
    \widetilde{A}(v,x) a = 0.
  \end{displaymath}
  It follows that under \eqref{eq:51} the matrices $K(v,x)$ and
  $\widetilde{A}(v,x)$ can have full rank only simultaneously.

  \ref{item:12} $\Longleftrightarrow$ \ref{item:14}. We fix arguments
  $(u,y)$ and $(v,x)$ satisfying the conjugacy relations of
  Lemma~\ref{lem:2}. From the definition of $\widetilde{B} =
  \widetilde{B}(u,y)$ in \eqref{eq:53} we deduce that the Hessian
  matrix of $g$ at $(u,y)$ has the representation
  \begin{displaymath}
    L(u,y) =
    \begin{pmatrix}
      \widetilde{B}(u,y)  & \frac{\partial g}{\partial u}(u,1) \\
      \transpose{(\frac{\partial g}{\partial u}(u,1))} & 0
    \end{pmatrix}
    =
    \begin{pmatrix}
      \widetilde{B}  & \frac1y v \\
      \frac1y \transpose{v} & 0
    \end{pmatrix}.
  \end{displaymath}

  To simplify notations we shall also represent the Hessian matrix of
  $f$ at $(v,x)$ as
  \begin{displaymath}
    K(v,x) =     \begin{pmatrix}
      M & p \\
      \transpose{p} & z
    \end{pmatrix},
  \end{displaymath}
  where $M$ is the Hessian matrix of $f(\cdot,x)$ at $v$, $p\set
  (\frac{\partial^2 f}{\partial v^m\partial x}(v,x))_{m=1,\dots,M}$ is
  the vector-column of mixed derivatives, and $z \set \frac{\partial^2
    f}{\partial x^2}(v,x)$. Observe that the matrix $\widetilde{A}
  =\widetilde{A}(v,x)$ defined in \eqref{eq:52} is given by
  \begin{equation}
    \label{eq:54}
    \widetilde{A} = M - \frac1z p \transpose{p}. 
  \end{equation}

  As $\widetilde{B}$ is a symmetric positive semi-definite matrix and
  $\frac1y v\not=0$, the full rank of $\widetilde{B}$ implies the full
  rank of $L(u,y)$.  Hence, by Lemma~\ref{lem:8}, under the conditions
  of either item~\ref{item:12} or item~\ref{item:14}, the Hessian
  matrices $K(v,x)$ and $L(u,y)$ have full rank and are inverse to
  each other. Denoting by $I$ the $M\times M$ identity matrix we
  deduce
  \begin{equation}
    \label{eq:55}
    \begin{split}
      \widetilde{B} M + \frac1y v \transpose{p} &= I, \\
      \widetilde{B} p + \frac1y v z & = 0.
    \end{split}
  \end{equation}

  If $z<0$, that is, \eqref{eq:51} holds, then, by \eqref{eq:54} and
  \eqref{eq:55}, $\widetilde{B}\widetilde{A} =I$. Hence,
  $\widetilde{B}$ is the inverse of $\widetilde{A}$ and, in
  particular, it has full rank, proving
  \ref{item:12}$\implies$\ref{item:14}.

  Conversely, if $\widetilde{B}$ has full rank, then $z =
  \frac{\partial^2 f}{\partial x^2}(v,x) \not=0$. Indeed, otherwise
  from the second equality in \eqref{eq:55} we obtain $p=0$
  contradicting the full rank of $K(v,x)$. Since, $f(v,\cdot)$ is
  concave, we deduce $z<0$, proving
  \ref{item:14}$\implies$\ref{item:12}.
\end{proof}

\begin{proof}[Proof of Theorem~\ref{th:2}.] If $(v,x,q)\in \mathbf{A}$
  and $(u,y,q)\in \mathbf{B}$ satisfy the equivalent relations of
  items~\ref{item:4}--\ref{item:7} of Theorem~\ref{th:1}, then the
  matrices $A$ and $B$ defined in \eqref{eq:14} and \eqref{eq:21} and
  the matrices $\widetilde A$ and $\widetilde B$ defined in
  \eqref{eq:52} and \eqref{eq:53} are related by
  \begin{align*}
    A^{kl} = \frac{v^k v^l}{y} \widetilde{A}^{kl}, \quad B^{kl} =
    \frac{y}{v^k v^l} \widetilde{B}^{kl}, \quad k,l=1,\dots,M.
  \end{align*}
  Lemma~\ref{lem:9} then implies \eqref{eq:48} as well as the other
  assertions of the theorem except those involving the second
  derivatives with respect to $q$.

  Assume first that $f\in \mathbf{F}^2$. We have to show that $g$ is
  two-times continuously differentiable and \eqref{eq:49} and
  \eqref{eq:50} hold. For $a\in \mathbf{R}^M$ define the function
  $\map{h}{(0,\infty)^M\times (-\infty,0)^M \times
    \mathbf{R}^J}{\mathbf{R}^M}$ by
  \begin{displaymath}
    h(v,u,q) \set (\frac{\partial f}{\partial v}(v,g(u,1,q),q)-u) + a
    (\frac{\partial f}{\partial x}(v,g(u,1,q),q) - 1).
  \end{displaymath}
  From Theorem~\ref{th:1} we deduce, for every $(u,y,q)\in
  \mathbf{B}$,
  \begin{displaymath}
    h(\frac{\partial g}{\partial u}(u,y,q),u,q) = 0.
  \end{displaymath}

  Fix $(u_0,y_0, q_0)\in \mathbf{B}$, denote $v_0 \set \frac{\partial
    g}{\partial u}(u_0,y_0,q_0)$, $x_0 \set g(u_0,1,q_0)$, and choose
  \begin{displaymath}
    a^m \set - \left(\frac{\partial^2 f}{\partial v^m \partial 
        x}/\frac{\partial^2 f}{\partial x^2}\right)(v_0,x_0,q_0), \quad
    m=1,\dots, M.  
  \end{displaymath}
  Direct computations show that, for $m,l=1,\dots,M$ and
  $j=1,\dots,J$,
  \begin{align*}
    \frac{\partial h^m}{\partial v^l}(v_0,u_0,q_0) &=
    \widetilde{A}^{ml}(v_0,x_0,q_0) = \frac{y_0}{v_0^m
      v_0^l}{A}^{ml}(v_0,x_0,q_0), \\
    \frac{\partial h^m}{\partial q^j}(v_0,u_0,q_0) & =
    \frac{y_0}{v^m_0}C^{mj}(v_0,x_0,q_0).
  \end{align*}
  By the Implicit Function Theorem the function $\frac{\partial
    g}{\partial u}=\frac{\partial g}{\partial u}(u,y,q)$ is
  continuously differentiable with respect to $q$ in a neighborhood of
  $(u_0,y_0,q_0)$ and the relation \eqref{eq:49} holds at this point.

  To prove the existence of the continuous second derivatives of $g$
  with respect to $q$ and the remaining identity \eqref{eq:50} we
  denote
  \begin{displaymath}
    b^j \set (\frac{\partial^2 f}{\partial x\partial
      q^j}/\frac{\partial^2 
      f}{\partial x^2})(v_0,x_0,q_0), \quad j=1,\dots, J. 
  \end{displaymath}
  From Theorem~\ref{th:1} we deduce, for every $(u,y,q)\in
  \mathbf{B}$,
  \begin{align*}
    \frac{\partial g}{\partial q}(u,y,q) + by = &- \frac{\partial
      f}{\partial q}(\frac{\partial g}{\partial
      u}(u,y,q),g(u,1,q),q) \\
    & + b \frac{\partial f}{\partial x}(\frac{\partial g}{\partial
      u}(u,y,q),g(u,1,q),q).
  \end{align*}
  This implies the two-times continuous differentiability of $g$ with
  respect to $q$. Moreover, direct computations show that the
  differentiation of the above identity with respect to $q$ at
  $(u_0,y_0,q_0)$ yields \eqref{eq:50} at this point.

  Assume now that $g\in \mathbf{G}^2$. To complete the proof we have
  to show that $f$ has continuous second derivatives involving $q$. By
  Theorem~\ref{th:1}, for every $(v,x,q)\in \mathbf{A}$ we have the
  equalities
  \begin{align*}
    \frac{\partial g}{\partial u}(\frac{\partial f}{\partial
      v}(v,x,q),\frac{\partial f}{\partial x}(v,x,q),q) - v
    &= 0, \\
    \frac{\partial g}{\partial y}(\frac{\partial f}{\partial
      v}(v,x,q),\frac{\partial f}{\partial x}(v,x,q),q) - x =
    g(\frac{\partial f}{\partial
      v}(v,x,q),1,q) - x & = 0, \\
    \frac{\partial f}{\partial q}(v,x,q) + \frac{\partial g}{\partial
      q}(\frac{\partial f}{\partial v}(v,x,q),\frac{\partial
      f}{\partial x}(v,x,q),q) &= 0.
  \end{align*}
  By Lemmas~\ref{lem:8} and~\ref{lem:9}, the full rank of the matrix
  $B(u,y,q)$ implies the full rank of the Hessian matrix of
  $g(\cdot,\cdot,q)$ at $(u,y)$. An application of the Implicit
  Function Theorem to the first two equalities above then leads to the
  continuous differentiability of $\frac{\partial f}{\partial v}$ and
  $\frac{\partial f}{\partial x}$ with respect to $q$. By the third
  identity, this implies the existence and the continuity of
  $\frac{\partial^2 f}{\partial q^i\partial q^j}$.
\end{proof}

\subsection{Stability under convergence}
\label{sec:stab-under-conv}

Let $m$ be a non-negative integer and $U$ be an open subset of
$\mathbf{R}^d$. Denote by $\mathbf{C}^m =
\mathbf{C}^m(U,\mathbf{R}^n)$ the Fr\'echet space of $m$-times
continuously differentiable maps $\map{f}{U}{\mathbf{R}^n}$ with the
topology generated by the semi-norms
\begin{displaymath}
  \norm{f}_{m,C} \set \sum_{0\leq \abs{k} \leq m} \sup_{x\in C}
  \abs{D^{{k}} f(x)}, 
\end{displaymath}
where $C$ is a compact subset of $U$, ${k} = (k_1,\dots,k_d)$ is a
multi-index of non-negative integers, $\abs{k} \set \sum_{i=1}^d k_i$,
and
\begin{equation}
  \label{eq:56}
  D^k \set 
  \frac{\partial^{|k|}}{\partial x_1^{k_1}\dots \partial x_d^{k_d}}. 
\end{equation}
In particular, for $m=0$, $D^{0}$ is the identity operator and
$\norm{f}_{0,C} \set \sup_{x\in C} |f(x)|$.

The following results show that the conjugacy relations between the
spaces $\mathbf{F}^i$ and $\mathbf{G}^i$ established in
Theorems~\ref{th:1} and~\ref{th:2} remain stable under
$\mathbf{C}^i$-convergence, $i=1,2$. These results will be used in the
proofs of our main Theorems~\ref{th:9} and \ref{th:10} to establish
that the stochastic fields $F$ and $G$ have versions which are RCLL
(right-continuous with left limits) in the time-variable $t$.

\begin{Theorem}
  \label{th:3}
  Let $(f_n)_{n\geq 1}$ and $f$ belong to $\mathbf{F}^1$ and
  $(g_n)_{n\geq 1}$ and $g$ be their conjugate counterparts from
  $\mathbf{G}^1$. Then $(f_n)_{n\geq 1}$ converges to $f$ in
  $\mathbf{C}^1(\mathbf{A})$ if and only if $(g_n)_{n\geq 1}$
  converges to $g$ in $\mathbf{C}^1(\mathbf{B})$.
\end{Theorem}

\begin{Theorem}
  \label{th:4}
  Let $(f_n)_{n\geq 1}$ and $f$ belong to $\mathbf{F}^2$ and
  $(g_n)_{n\geq 1}$ and $g$ be their conjugate counterparts from
  $\mathbf{G}^2$. Then $(f_n)_{n\geq 1}$ converges to $f$ in
  $\mathbf{C}^2(\mathbf{A})$ if and only if $(g_n)_{n\geq 1}$
  converges to $g$ in $\mathbf{C}^2(\mathbf{B})$.
\end{Theorem}

\subsubsection{Proofs of Theorems~\ref{th:3} and \ref{th:4}}
\label{sec:proofs-theor-3}

\begin{proof}[Proof of Theorem~\ref{th:3}.]
  Recall that for convex or saddle functions the convergence in
  $\mathbf{C}^1$ is equivalent to the pointwise convergence. We also
  remind the reader that the conjugacy operations, as in~\eqref{eq:22}
  and~\eqref{eq:23}, are, in general, not continuous under this
  convergence and, hence, the result does not hold automatically.  A
  standard verification method in this case is to show the equivalence
  of the pointwise convergence and the \emph{epi}-convergence (or its
  analogs such as \emph{epi}-\emph{hypo}-convergence), under which the
  conjugacy operations are continuous; see \citet{RockWets:98},
  Theorem 11.34. We find it simpler to give a direct argument.

  Assume first that $(f_n)_{n\geq 1}$ converges to $f$ in 
  $\mathbf{C}^1(\mathbf{A})$.  By the positive homogeneity condition
  \ref{item:G4} and because they are saddle functions, it is
  sufficient to verify the pointwise convergence for $(g_n)_{n\geq 1}$
  at $b=(u,y,q)\in \mathbf{B}$ with $y=1$. Fix $\epsilon>0$ and find
  $u_i \in (-\infty,0)^M$, $i=1,2$, such that $u_1 < u < u_2$ and
  \begin{equation}
    \label{eq:57}
    |g(b_2) - g(b_1)| < \epsilon,  
  \end{equation}
  where $b_i \set (u_i,1,q)$.  Denote, for $i=1,2$,
  \begin{displaymath}
    a_i = (v_i, x_i, q) \set (\frac{\partial g}{\partial u}(b_i),
    g(b_i),q), 
  \end{displaymath}
  and, for $n\geq 1$,
  \begin{displaymath}
    b_{i,n} = (u_{i,n},1,q) \set (\frac{\partial f_n}{\partial v}(a_i),1,q). 
  \end{displaymath}

  The conjugacy relations between $f_n$ and $g_n$ and between $f$ and
  $g$ imply that $g_{n}(b_{i,n}) = x_i$ and $u_i = \frac{\partial
    f}{\partial v}(a_i)$.  From the $\mathbf{C}^1$-convergence of
  $(f_n)_{n\geq 1}$ to $f$ we deduce
  \begin{displaymath}
    \lim_{n\to\infty} u_{i,n} = \lim_{n\to\infty} \frac{\partial f_n}{\partial
      v}(a_i) = \frac{\partial f}{\partial v}(a_i) = u_i, \quad i=1,2, 
  \end{displaymath}
  and, hence, there is $n_0>1$ such that $u_{1,n} < u < u_{2,n}$ for
  $n\geq n_0$.  Accounting for the monotonicity of the elements of
  $\mathbf{G}^1$ with respect to $u$ we obtain
  \begin{gather*}
    g(b_1) < g(b) < g(b_2), \\
    g(b_1) = g_n(b_{1,n}) < g_n(b) < g_n(b_{2,n}) = g(b_2), \quad
    n\geq n_0,
  \end{gather*}
  and then \eqref{eq:57} yields
  \begin{displaymath}
    |g_n(b) - g(b)| < \epsilon, \quad n\geq n_0, 
  \end{displaymath}
  thus proving the pointwise, hence, also the
  $\mathbf{C}^1(\mathbf{B})$, convergence of $(g_n)_{n\geq 1}$ to $g$.

  Assume now that $(g_n)_{n\geq 1}$ converges to $g$ in
  $\mathbf{C}^1(\mathbf{B})$. We follow the same path as in the proof
  of the previous implication.  Fix $\epsilon>0$, take $a=(v,x,q)\in
  \mathbf{A}$ and let $a_i = (v_i,x_i,q) \in \mathbf{A}$, $i=1,2$, be
  such that $v_1 > v > v_2$, $x_1 < x < x_2$, and
  \begin{equation}
    \label{eq:58}
    |f(a_2) - f(a_1)| < \epsilon. 
  \end{equation}
  Denote, for $i=1,2$,
  \begin{displaymath}
    b_i = (u_i,y_i,q) \set (\frac{\partial f}{\partial v}(a_i), 
    \frac{\partial f}{\partial x}(a_i),q), 
  \end{displaymath}
  and, for $n\geq 1$,
  \begin{displaymath}
    a_{i,n} = (v_{i,n},x_{i,n},q) \set (\frac{\partial g_n}{\partial
      u}(b_i),\frac{\partial g_n}{\partial y}(b_i), q).  
  \end{displaymath}

  From the conjugacy relations between $f_n$ and $g_n$ and between $f$
  and $g$ we deduce that $f_{n}(a_{i,n}) = \ip{u_i}{v_{i,n}}$,
  $f(a_{i}) = \ip{u_i}{v_{i}}$, and $a_i = (\frac{\partial g}{\partial
    u}(b_i),\frac{\partial g}{\partial y}(b_i), q)$.  As the
  $\mathbf{C}^1$-convergence of $(g_n)_{n\geq 1}$ to $g$ implies the
  convergence of $(a_{i,n})_{n\geq 1}$ to $a_i$, there is $n_0>1$ such
  that, for $n\geq n_0$, $ v_{1,n} > v > v_{2,n}$, $x_{1,n} < x <
  x_{2,n}$, and $|\ip{u_i}{v_{i,n}} -
  \ip{u_i}{v_i}|<\epsilon$. Accounting for the monotonicity of the
  elements of $\mathbf{F}^1$ with respect to $v$ and $x$, we deduce
  \begin{gather*}
    f(a_1) < f(a) < f(a_2), \\
    f(a_1) - \epsilon < f_n(a_{1,n}) < f_n(a) < f_n(a_{2,n}) < f(a_2)
    + \epsilon, \quad n\geq n_0,
  \end{gather*}
  and, then, \eqref{eq:58} implies
  \begin{displaymath}
    |f_n(a) - f(a)| < 2\epsilon, \quad n\geq n_0, 
  \end{displaymath}
  proving the pointwise (hence, also $\mathbf{C}^1(\mathbf{A})$)
  convergence of $(f_n)_{n\geq 1}$ to $f$.
\end{proof}

For the proof of Theorem~\ref{th:4} we need some elementary identities
for the matrices $A(f)$ and $C(f)$.

\begin{Lemma}
  \label{lem:10}
  For $f\in \mathbf{F}^2$, the matrix $A(f)$ defined in \eqref{eq:14}
  satisfies
  \begin{align*}
    \sum_{m=1}^M A^{lm}(f) &= - v^l\frac{\partial^2 f}{\partial
      v^l \partial x}/\frac{\partial^2 f}{\partial x^2}, \quad
    l=1,\dots,M,
    \\
    \sum_{l,m=1}^M A^{lm}(f) &= -\frac{\partial f}{\partial
      x}/\frac{\partial^2 f}{\partial x^2}.
  \end{align*}
\end{Lemma}

\begin{proof}
  From the positive homogeneity condition \eqref{eq:11} we deduce
  \begin{align*}
    \sum_{m=1}^M v^m \frac{\partial^2 f}{\partial v^l \partial v^m} &=0, \\
    \sum_{m=1}^M v^m \frac{\partial^2 f}{\partial x\partial v^m}
    &=\frac{\partial f}{\partial x},
  \end{align*}
  and the result follows.
\end{proof}

\begin{Lemma}
  \label{lem:11}
  For $f\in \mathbf{F}^2$, the matrix $C(f)$ defined in~\eqref{eq:44}
  satisfies
  \begin{equation}
    \label{eq:59}
    \sum_{m=1}^M C^{mj}(f)  = \frac{1}{\frac{\partial f}{\partial x}}
      \frac{\partial f}{\partial q^j} - \frac{1}{\frac{\partial^2 f}{\partial x^2}}
      \frac{\partial^2 f}{\partial q^j \partial x}, \quad
    j=1,\dots,J. 
  \end{equation}
\end{Lemma}
\begin{proof} The positive homogeneity property~\eqref{eq:11} yields
  Euler's identity:
  \begin{displaymath}
    f = \sum_{m=1}^M v^m \frac{\partial f}{\partial v^m}, 
  \end{displaymath}
  which, in turn, implies
  \begin{align*}
    \frac{\partial f}{\partial x} &= \sum_{m=1}^M v^m \frac{\partial^2
      f}{\partial v^m \partial x}, \\
    \frac{\partial f}{\partial q^j}&= \sum_{m=1}^M v^m
    \frac{\partial^2 f}{\partial v^m \partial q^j},\quad j=1,\dots,J,
  \end{align*}
  proving \eqref{eq:59}.
\end{proof}

\begin{proof}[Proof of Theorem~\ref{th:4}]
  In view of Theorem~\ref{th:3}, we only have to establish the uniform
  on compact sets convergences of second derivatives. Recall the
  notations $A(f)$, $C(f)$, and $D(f)$, for the matrices defined in
  \eqref{eq:14}, \eqref{eq:44} and \eqref{eq:45}, and $B(g)$, $E(g)$,
  and $H(g)$, for the matrices defined in~\eqref{eq:21},
  \eqref{eq:46}, and~\eqref{eq:47}.

  Assume first that $(f_n)_{n\geq 1}$ converges to $f$ in
  $\mathbf{C}^2(\mathbf{A})$.  Let $(b_n)_{n\geq 1}$ be a sequence in
  $\mathbf{B}$ that converges to $b\in \mathbf{B}$. By
  Theorem~\ref{th:3}, the sequence $a_n \set (\frac{\partial
    g_n}{\partial u}(b_n), \frac{\partial g_n}{\partial y}(b_n),q_n)$,
  $n\geq 1$, converges to $a \set (\frac{\partial g}{\partial u}(b),
  \frac{\partial g}{\partial y}(b),q)$. The convergence of
  $(f_n)_{n\geq 1}$ to $f$ in $\mathbf{C}^2(\mathbf{A})$ then implies
  the convergence of the matrices
  $((A(f_n),C(f_n),D(f_n))(a_n))_{n\geq 1}$, to $(A(f),C(f),D(f))(a)$.
  By the identities~\eqref{eq:48}, \eqref{eq:49}, and~\eqref{eq:50},
  this implies the convergence of the matrices
  $((B,E,H)(g_n)(b_n))_{n\geq 1}$, to $(B,E,H)(g)(b)$, which, by the
  construction of these matrices, yields the convergence of all second
  derivatives of $g_n$ at $b_n$, $n\geq 1$, to the corresponding
  second derivatives of $g$ at $b$. This, clearly, implies the uniform
  on compact sets convergence of the second derivatives of
  $(g_n)_{n\geq 1}$ to $g$.

  Similar arguments show that the
  $\mathbf{C}^2(\mathbf{B})$-convergence of $(g_n)_{n\geq 1}$ to $g$
  implies that for every sequence $(a_n)_{n\geq 1}$ in $\mathbf{A}$
  converging to $a\in \mathbf{A}$ the matrices
  $((A(f_n),C(f_n),D(f_n))(a_n))_{n\geq 1}$ converge to
  $(A(f),C(f),D(f))(a)$. This, in turn, implies the convergence of the
  second derivatives of $f_n$ at $a_n$, $n\geq 1$, to the second
  derivatives of $f$ at $a$ if we account for the
  identity~\eqref{eq:59} for the matrix $C(f)$ and the equalities for
  the matrix $A(f)$ from Lemma~\ref{lem:10}.
\end{proof}

\subsection{Additional conjugacy relations}
\label{sec:addit-conj-relat}

If $f\in \mathbf{F}^1$ and $g\in \mathbf{G}^1$ are conjugate in the
sense that~\eqref{eq:22} and \eqref{eq:23} hold true, then any extra
condition for $f$ has its conjugate analog for $g$. Below we shall
present several such extensions, which will appear in the description
of the sample paths of the stochastic fields $F$ and $G$. Throughout
this section we fix a constant $c>0$.

For a function $\map{f}{\mathbf{A}}{(-\infty,0)}$ define the following
conditions:

\begin{enumerate}[label=(F\arabic{*}), ref=(F\arabic{*})]
  \setcounter{enumi}{5}
\item \label{item:F6} If $M>1$, then for every $(x,q)\in \mathbf{R}
  \times \mathbf{R}^J$ and every sequence $(w_n)_{n\geq 1}$ in
  $\mathbf{S}^M$ converging to a boundary point of $\mathbf{S}^M$ we
  have
  \begin{displaymath}
    \lim_{n\to \infty} \sum_{m=1}^M \frac{\partial f}{\partial v^m}(w_n,x,q)
    = -\infty.
  \end{displaymath}
\item \label{item:F7} For every $a=(v,x,q)\in \mathbf{A}$ and
  $m=1,\dots,M$,
  \begin{displaymath}
    \frac1c\frac{\partial f}{\partial
      x}(a)  \leq -v^m \frac{\partial f}{\partial v^m}(a) \leq c
    \frac{\partial f}{\partial x}(a).
  \end{displaymath}
\item \label{item:F8} For every $a\in \mathbf{A}$ and every $z\in
  \mathbf{R}^M$,
  \begin{displaymath}
    \frac1{c} \ip{z}{z} \leq \ip{z}{{A}(f)(a)z} \leq c \ip{z}{z}, 
  \end{displaymath}
  where the matrix $A(f)(a)$ is defined in \eqref{eq:14}.
\item \label{item:F9} For every $a=(v,x,q)\in \mathbf{A}$ and
  $m=1,\dots,M$,
  \begin{displaymath}
    -\frac1c\frac{\partial^2 f}{\partial x^2}(a)
    \leq v^m \frac{\partial^2 f}{\partial v^m \partial x}(a) \leq -c 
    \frac{\partial^2 f}{\partial x^2}(a).
  \end{displaymath}
\end{enumerate}

For a function $\map{g}{\mathbf{B}}{\mathbf{R}}$ define the following
conditions:

\begin{enumerate}[label=(G\arabic{*}), ref=(G\arabic{*})]
  \setcounter{enumi}{5}
\item \label{item:G6} For every $(y,q)\in (0,\infty) \times
  \mathbf{R}^J$ and every sequence $(u_n)_{n\geq 1}$ in
  $(-\infty,0)^M$ converging to a boundary point of $(-\infty,0)^M$ we
  have
  \begin{displaymath}
    \lim_{n\to \infty} g(u_n,y,q) = \infty.
  \end{displaymath}
\item \label{item:G7} For every $(u,q)\in (-\infty,0)^M\times
  \mathbf{R}^J$ and $m=1,\dots,M$,
  \begin{displaymath}
    \frac1c \leq -u^m \frac{\partial g}{\partial u^m}(u,1,q) \leq c. 
  \end{displaymath}
\item \label{item:G8} For every $b\in \mathbf{B}$ and every $z\in
  \mathbf{R}^M$,
  \begin{displaymath}
    \frac1{c} \ip{z}{z} \leq \ip{z}{{B}(g)(b)z} \leq c \ip{z}{z}, 
  \end{displaymath}
  where the matrix $B(g)(b)$ is defined in \eqref{eq:21}.
\item \label{item:G9} For every $(u,q)\in (-\infty,0)^M\times
  \mathbf{R}^J$, the vector $z\in \mathbf{R}^M$ solving the linear
  equation:
  \begin{displaymath}
    B(g)(u,1,q) z = \idvec,
  \end{displaymath}
  where $\idvec\set (1,\dots,1) \in \mathbf{R}^M$, satisfies
  \begin{displaymath}
    \frac1c \leq z^m \leq c, \quad m=1,\dots, M. 
  \end{displaymath}
\end{enumerate}

Note that for $g\in \mathbf{G}^1$ the condition \ref{item:G6} implies
(a) and (b) in \ref{item:G2} and holds trivially when $M=1$
by~\eqref{eq:15}.

\begin{Theorem}
  \label{th:5}
  Let $f\in \mathbf{F}^1$ and $g\in \mathbf{G}^1$ be conjugate in the
  sense of \eqref{eq:22} and \eqref{eq:23}. Then \ref{item:F6} is
  equivalent to \ref{item:G6} and \ref{item:F7} is equivalent to
  \ref{item:G7}.  If, in addition, $f\in \mathbf{F}^2$ and (hence)
  $g\in \mathbf{G}^2$ then \ref{item:F8} is equivalent to
  \ref{item:G8} and \ref{item:F9} is equivalent to \ref{item:G9}.
\end{Theorem}

\subsubsection{Proof of Theorem \ref{th:5}}
\label{sec:proof-theor-refth:11}

The proof follows from the lemmas below, where $f\in \mathbf{F}^1$ and
$g\in \mathbf{G}^1$ are conjugate as in \eqref{eq:22}.

\begin{Lemma}
  \label{lem:12}
  Suppose $M>1$. Then the conditions \ref{item:F6} and \ref{item:G6}
  are equivalent.
\end{Lemma}

\begin{proof} To simplify notations we shall omit the dependence of
  $f$ and $g$ on the irrelevant parameter $q$. Recall the notations
  $\boundary{A}$ and $\closure{A}$ for the boundary and the closure of
  a set $A$.

  \ref{item:F6} $\implies$ \ref{item:G6}. Let $(u_n)_{n\geq 1}$ be a
  sequence in $(-\infty,0)^M$ converging to $\widehat u\in
  \boundary{(-\infty,0)^M}$. Denote $x_n \set g(u_n,1)$, $n\geq 1$,
  and, contrary to \ref{item:G6}, suppose
  \begin{displaymath}
    \liminf_{n\to \infty} g(u_n,1) = \liminf_{n\to \infty} x_n < \infty. 
  \end{displaymath}
  As $g(\cdot,1)$ is an increasing function on $(-\infty,0)^M$ and the
  sequence $(u_n)_{n\geq 1}$ is bounded from below, the sequence
  $(x_n)_{n\geq 1}$ is also bounded from below. Hence, by passing to a
  subsequence, we can assume that $(x_n)_{n\geq 1}$ converges to some
  $\widehat x \in \mathbf{R}$.

  Denoting $v_n \set \frac{\partial g}{\partial u}(u_n,1)$ and $w_n
  \set \frac{v_n}{\sum_{m=1}^M v_n^m}$ we deduce from
  items~\ref{item:6} and \ref{item:7} of Theorem~\ref{th:1} and the
  positive homogeneity condition \eqref{eq:11} for $f$ that
  \begin{displaymath}
    u_n = \frac{\partial f}{\partial v}(v_n,x_n) = \frac{\partial
      f}{\partial v}(w_n,x_n), \quad n\geq 1.  
  \end{displaymath}
  As $w_n\in \mathbf{S}^M$, passing to a subsequence, we can assume
  that $(w_n)_{n\geq 1}$ converges to $\widehat
  w\in\closure{\mathbf{S}^M}$.  If $\widehat w\in \mathbf{S}^M$, then
  \begin{displaymath}
    \widehat u =  \lim_{n\to \infty} u_n = \lim_{n\to \infty}
    \frac{\partial f}{\partial 
      v}(w_n,x_n) = \frac{\partial f}{\partial v}(\widehat w,\widehat
    x) \in (-\infty,0)^M,  
  \end{displaymath}
  contradicting our choice of $\widehat u$.  If, on the other hand,
  $\widehat w\in \boundary{\mathbf{S}^M}$, then, by \eqref{eq:12} and
  the monotonicity of $f=f(v,x)$ with respect to $x$,
  \begin{displaymath}
    \lim_{n\to \infty} f(w_n,x_n) = 0. 
  \end{displaymath}

  It follows that, for every $v\in (0,\infty)^M$,
  \begin{gather*}
    0 \geq f(v,\widehat x) = \lim_{n\to \infty} f(v,x_n) \geq
    \lim_{n\to \infty}\left(f(w_n,x_n) + \ip{\frac{\partial
          f}{\partial v}(w_n,x_n)}{v-w_n}\right) \\ = \lim_{n\to
      \infty}(f(w_n,x_n) + \ip{u_n}{v-w_n}) = \ip{\widehat
      u}{v-\widehat w}.
  \end{gather*}
  Hence, if we extend, by continuity, the convex function
  $f(\cdot,\widehat x)$ to the boundary of its domain by setting
  \begin{displaymath}
    f(v,\widehat x) = 0, \quad v \in \boundary{(0,\infty)^M},
  \end{displaymath}
  then its subdifferential $\partial f_v(\widehat w,\widehat x)$ at
  $\widehat w$ contains $\widehat u$ and, therefore, is non-empty. In
  this case, there are $\widetilde u \in \partial f_v(\widehat
  w,\widehat x)$ and a sequence $(\widetilde v_n)_{n\geq 1} \subset
  (0,\infty)^M$ convergent to $\widehat w$ such that
  \begin{displaymath}
    \widetilde u =  \lim_{n\to \infty}  \frac{\partial f}{\partial
      v}(\widetilde v_n,\widehat x),
  \end{displaymath}
  see Theorem 25.6 in \cite{Rock:70}. Denoting $\widetilde w_n \set
  \frac{\widetilde v_n}{\sum_{m=1}^M \widetilde v_n^m}$, $n\geq 1$,
  and accounting for \eqref{eq:11} we obtain
  \begin{displaymath}
    \widetilde u =  \lim_{n\to \infty}  \frac{\partial f}{\partial
      v}(\widetilde w_n,\widehat x),
  \end{displaymath}
  which contradicts \ref{item:F6}.

  \ref{item:G6} $\implies$ \ref{item:F6}. Fix $x\in\mathbf{R}$, let
  $(w_n)_{n\geq 1}$ be a sequence in $\mathbf{S}^M$ converging to
  $w\in \boundary{\mathbf{S}^M}$, and denote
  \begin{displaymath}
    u_n \set \frac{\partial f}{\partial v}(w_n,x), \quad y_n \set
    \frac{\partial f}{\partial x}(w_n,x), \quad n\geq 1. 
  \end{displaymath}
  By \eqref{eq:12}, the concave functions $f(w_n,\cdot)$, $n\geq 1$,
  converge to $0$. Hence, their derivatives also converge to $0$,
  implying that
  \begin{equation}
    \label{eq:60}
    \lim_{n\to\infty} y_n =0.
  \end{equation}
  Contrary to \ref{item:F6} suppose $(u_n)_{n\geq 1}$ contains a
  bounded subsequence. Passing to a subsequence we can then assume
  that $(u_n)_{n\geq 1}$ converges to $u\in (-\infty,0]^M$.

  By the equivalence of items~\ref{item:6} and \ref{item:7} in
  Theorem~\ref{th:1},
  \begin{displaymath}
    w_n = y_n \frac{\partial g}{\partial u}(u_n,1), \quad x=g(u_n,1), \quad
    n\geq 1. 
  \end{displaymath}
  In view of \eqref{eq:60}, the first equality implies that $u\not\in
  (-\infty,0)^M$. Hence, $u\in \boundary{(-\infty,0)^M}$,
  contradicting \ref{item:G6} and the second equality.
\end{proof}

Recall that a constant $c>0$ appearing in \ref{item:F7}--\ref{item:F9}
and \ref{item:G7}--\ref{item:G9} is fixed.

\begin{Lemma}
  \label{lem:13}
  The conditions \ref{item:F7} and \ref{item:G7} are equivalent.
\end{Lemma}

\begin{proof}
  Follows from the items~\ref{item:6} and \ref{item:7} in the list of
  equivalent characterizations of saddle points in Theorem~\ref{th:1}
  and the positive homogeneity condition \eqref{eq:20} for $g$.
\end{proof}

From now on we assume, in addition, that $f\in \mathbf{F}^2$ and $g\in
\mathbf{G}^2$.

\begin{Lemma}
  \label{lem:14}
  The conditions \ref{item:F8} and \ref{item:G8} are equivalent.
\end{Lemma}
\begin{proof}
  Follows from the inverse relation \eqref{eq:48} between the matrices
  $A(f)(a)$ and $B(g)(b)$.
\end{proof}

\begin{Lemma}
  \label{lem:15}
  The conditions \ref{item:F9} and \ref{item:G9} are equivalent.
\end{Lemma}

\begin{proof}
  By Lemma~\ref{lem:10}, the condition \ref{item:F9} can be
  equivalently stated as
  \begin{displaymath}
    \frac1c \leq (A(f)(a)\idvec)^m \leq c, \quad m=1,\dots, M, 
  \end{displaymath}
  and the result follows from the inverse relation \eqref{eq:48}
  between the matrices $A(f)(a)$ and $B(g)(b)$.
\end{proof}

\subsection{The spaces $\widetilde{\mathbf{F}}^1$,
  $\widetilde{\mathbf{G}}^1$ and $\widetilde{\mathbf{F}}^2(c)$,
  $\widetilde{\mathbf{G}}^2(c)$}
\label{sec:spaces-tilde}

To simplify future references we define the following families of functions:
\begin{align*}
  \widetilde{\mathbf{F}}^1 \set \descr{f\in
    \mathbf{F}^1}{\text{\ref{item:F6} holds}}, \\
  \widetilde{\mathbf{G}}^1 \set \descr{g\in
    \mathbf{G}^1}{\text{\ref{item:G6} holds}},
\end{align*}
and, for a constant $c>0$,
\begin{align*}
  \widetilde{\mathbf{F}}^2(c) \set \descr{f\in \mathbf{F}^2 \cap
    \widetilde{\mathbf{F}}^1}{\text{\ref{item:F7}--\ref{item:F9} hold
      for given $c$}}, \\
  \widetilde{\mathbf{G}}^2(c) \set \descr{g\in \mathbf{G}^2\cap
    \widetilde{\mathbf{G}}^1}{\text{\ref{item:G7}--\ref{item:G9} hold
      for given $c$}}.
\end{align*}
Note that when $M=1$ the conditions~\ref{item:F6} and~\ref{item:G6}
hold trivially; in particular, $\widetilde{\mathbf{F}}^1=\mathbf{F}^1$
and $\widetilde{\mathbf{G}}^1= \mathbf{G}^1$.

From Theorems~\ref{th:1}, ~\ref{th:2}, and~\ref{th:5} we immediately
obtain

\begin{Theorem}
  \label{th:6}
  A function $\map{f}{\mathbf{A}}{(-\infty,0)}$ belongs to
  $\widetilde{\mathbf{F}}^1$ if and only if it is conjugate to a
  function $g\in \widetilde{\mathbf{G}}^1$ in the sense of
  \eqref{eq:22} and \eqref{eq:23}. Moreover, if $c>0$, then $f\in
  \widetilde{\mathbf{F}}^2(c)$ if and only if $g\in
  \widetilde{\mathbf{G}}^2(c)$.
\end{Theorem}

\section{Aggregate utility function}
\label{sec:repr-mark-maker}

Recall that the aggregate utility function $r=r(v,x)$ is given by
\begin{equation}
  \label{eq:61}
  r(v,x) \set \sup_{x^1 + \dots + x^M = x} \sum_{m=1}^M v^m u_m(x^m), \quad
  v\in (0,\infty)^M, x\in \mathbf{R}.
\end{equation}
Theorems~\ref{th:7} and~\ref{th:8} below identify $r=r(v,x)$ as an
element of $\widetilde{\mathbf{F}}^1$ and
$\widetilde{\mathbf{F}}^2(c)$ under Assumptions~\ref{as:1} and
\ref{as:2}, respectively. Note that throughout this section we
interpret these families of functions, defined in
Section~\ref{sec:spaces-tilde}, in the sense of Remark~\ref{rem:1}.

\begin{Theorem}
  \label{th:7}
  Under Assumption~\ref{as:1} the function $r = r(v,x)$ belongs to
  $\widetilde{\mathbf{F}}^1$. Moreover, for every $(v,x)\in
  (0,\infty)^M \times \mathbf{R}$, the supremum in \eqref{eq:61} is
  attained at the vector $\widehat{x} \in \mathbf{R}^M$ uniquely
  determined by \eqref{eq:62} or, equivalently, \eqref{eq:63} below:
  \begin{align}
    \label{eq:62}
    v^m u'_m(\widehat{x}^m) &= \frac{\partial r}{\partial x} (v,x), \\
    \label{eq:63}
    u_m(\widehat{x}^m) & = \frac{\partial r}{\partial v^m} (v,x),
    \quad m=1, \dots, M.
  \end{align}
\end{Theorem}

Denote by $t_m = t_m(x)$ the risk-tolerance coefficients of the
utility functions $u_m = u_m(x)$:
\begin{equation}
  \label{eq:64}
  t_m(x) \set - \frac{u'_m(x)}{u''_m(x)}=\frac1{a_m(x)}, \quad x\in
  \mathbf{R},\; m=1,\dots,M.
\end{equation}
Hereafter, the symbol $\widehat x = \widehat x(v,x)$ is used to define
the functional dependence of the maximal vector $\widehat x =
(\widehat x^m)_{m=1,\dots,M}$ from Theorem~\ref{th:7} on $v$ and $x$.

\begin{Theorem}
  \label{th:8}
  Under Assumptions~\ref{as:1} and~\ref{as:2} the function $r =
  r(v,x)$ belongs to $\widetilde{\mathbf{F}}^2(c)$ with the same
  constant $c>0$ as in \eqref{eq:3}, the function $\widehat x =
  \widehat x(v,x)$ is continuously differentiable, and, for
  $l,m=1,\dots,M$,
  \begin{align}
    \label{eq:65}
    \frac{\partial \widehat x^m}{\partial x}(v,x) & =
    \frac{t_m(\widehat
      x^m)}{\sum_{k=1}^M t_k(\widehat x^k)}, \\
    \label{eq:66}
    v^l \frac{\partial \widehat x^m}{\partial v^l}(v,x) & = v^m
    \frac{\partial \widehat x^l}{\partial v^m}(v,x) = t_m(\widehat
    x^m)\left(\delta_{lm} - \frac{t_l(\widehat x^l)}{\sum_{k=1}^M
        t_k(\widehat x^k)}\right),
  \end{align}
  where $\delta_{lm} \set \ind{l=m}$ is the Kronecker delta,
  \begin{align}
    \label{eq:67}
    \frac{\partial^2 r}{\partial x^2}(v,x) &= -\frac{\partial
      r}{\partial
      x}(v,x) \frac1{\sum_{k=1}^M t_k(\widehat x^k)}, \\
    \label{eq:68}
    v^m \frac{\partial^2 r}{\partial v^m\partial x}(v,x) &=
    \frac{\partial r}{\partial
      x}(v,x) \frac{t_m(\widehat x^m)}{\sum_{k=1}^M t_k(\widehat x^k)}, \\
    \label{eq:69}
    v^l v^m \frac{\partial^2 r}{\partial v^l \partial v^m}(v,x) &=
    \frac{\partial r}{\partial x}(v,x) t_l(\widehat
    x^l)\left(\delta_{lm} - \frac{t_m(\widehat x^m)}{\sum_{k=1}^M
        t_k(\widehat x^k)}\right),
  \end{align}
  and, for the matrix $A(r)$ in \ref{item:F5},
  \begin{equation}
    \label{eq:70}
    \begin{split}
      A^{lm}(r)(v,x) &\set \frac{v^lv^m}{\frac{\partial r}{\partial
          x}} \left(\frac{\partial^2 r}{\partial v^l\partial v^m} -
        \frac1{\frac{\partial^2 r}{\partial x^2}} \frac{\partial^2
          r}{\partial v^l\partial x}\frac{\partial^2 r}{\partial
          v^m\partial
          x}\right)(v,x) \\
      &= t_l(\widehat x^l) \delta_{lm}.
    \end{split}
  \end{equation}
\end{Theorem}

\subsection{Proofs of Theorems \ref{th:7} and \ref{th:8}}
\label{sec:proofs-theor-refth:6}

\begin{proof}[Proof of Theorem~\ref{th:7}.]
  Define the function $\map{g=g(v,x,z)}{(0,\infty)^M \times \mathbf{R}
    \times \mathbf{R}^{M-1}}{\mathbf{R}}$ by
  \begin{displaymath}
    g(v,x,z) \set \sum_{m=1}^{M-1} v^m u_m(z^m) + v^M u_M(x -
    \sum_{m=1}^{M-1} z^m) 
  \end{displaymath}
  and observe that
  \begin{equation}
    \label{eq:71}
    r(v,x) = \sup_{z\in \mathbf{R}^{M-1}} g(v,x,z), \quad v\in (0,\infty)^M,
    x\in \mathbf{R}. 
  \end{equation}
  For every $v\in (0,\infty)^M$, the function $g(v,\cdot,\cdot)$ is
  strictly concave, continuously differentiable, and, by \eqref{eq:1}
  and \eqref{eq:2}, for every $x\in \mathbf{R}$,
  \begin{displaymath}
    \lim_{\abs{z}\to \infty} g(v,x,z) = - \infty. 
  \end{displaymath}
  It follows that the upper bound in \eqref{eq:71} is attained at a
  unique $\widehat z=\widehat z(v,x)$ satisfying, for $m=1,\dots,M-1$,
  \begin{equation}
    \label{eq:72}
    0 = \frac{\partial g}{\partial z^m}(v,x,\widehat z) =
    v^mu'_m(\widehat{z}^m) - v^M u_M'(x - \sum_{k=1}^{M-1}
    \widehat z^k).
  \end{equation}
  Hence, the upper bound in \eqref{eq:61} is attained at the unique
  $\widehat x = (\widehat{x}^m)_{m=1,\dots,M}$ given by
  \begin{align*}
    \widehat{x}^m &= \widehat{z}^m, \quad m=1,\dots,M-1, \\
    \widehat{x}^M &= x - \sum_{m=1}^{M-1} \widehat z^m.
  \end{align*}

  By Lemma~\ref{lem:20} in Appendix~\ref{sec:envelope-theorem} the
  function $r(v,\cdot)$ is concave, differentiable, (hence,
  continuously differentiable), and
  \begin{displaymath}
    \frac{\partial r}{\partial x}(v,x) = \frac{\partial g}{\partial
      x}(v,x,\widehat z) = v^M u_M'(x - \sum_{m=1}^{M-1}
    \widehat z^k)  = v^M u'_M(\widehat x^M),
  \end{displaymath}
  which jointly with \eqref{eq:72} proves \eqref{eq:62}. As $u'_M>0$
  we have $\frac{\partial r}{\partial x}>0$ and, hence, $r(v,\cdot)$
  is strictly increasing. From \eqref{eq:1} we obtain
  \begin{displaymath}
    \lim_{x\to\infty} r(v,x)  = 0. 
  \end{displaymath}
  Finally, the strict concavity of $r(v,\cdot)$ follows directly from
  the strict concavity of $(u_m)_{m=1,\dots,M}$ and the attainability
  of the upper bound in \eqref{eq:61}, thus finishing the verification
  of \ref{item:F4}.

  For $(x,z) \in \mathbf{R}^M$, the function $g(\cdot,x,z)$ is affine
  on $(0,\infty)^M$ and, in particular, convex and continuously
  differentiable.  Hence, by Lemma~\ref{lem:21} in
  Appendix~\ref{sec:envelope-theorem}, the function $r(\cdot,x)$ is
  convex, differentiable, (hence, continuously differentiable), and
  \begin{displaymath}
    \frac{\partial r}{\partial v^m}(v,x) = \frac{\partial g}{\partial
      v^m}(v,x,\widehat z) = u_m(\widehat x^m), \quad m=1,\dots,M, 
  \end{displaymath}
  proving \eqref{eq:63}.  As $u_m<0$, the function $r(\cdot,x)$ is
  strictly decreasing. It is, clearly, positively
  homogeneous. Moreover, if $M>1$ then by \eqref{eq:1}
  \begin{equation}
    \label{eq:73}
    \lim_{n\to\infty} r(w_n,x) = 0, 
  \end{equation}
  for every sequence $(w_n)_{n\geq 1}$ in $\mathbf{S}^M$ converging to
  $w\in \boundary{\mathbf{S}^M}$. Hence, to complete the verification
  of \ref{item:F2} we only need to show the strict convexity of this
  function on $\mathbf{S}^M$.

  Let $w_1$ and $w_2$ be distinct elements of $\mathbf{S}^M$, $w_3$ be
  their midpoint, and $\widehat x_i$ be the points in $\mathbf{R}^M$
  where the upper bound in \eqref{eq:61} is attained for $r(x,w_i)$,
  $i=1,2,3$. From \eqref{eq:62} we deduce that the points $(\widehat
  x_i)_{i=1,2,3}$ are distinct and, hence,
  \begin{align*}
    r(w_3,x) &= \sum_{m=1}^M w_3^m u_m(\widehat x^m_3) =
    \frac12\left(\sum_{m=1}^M w_1^m u_m(\widehat x^m_3) + \sum_{m=1}^M
      w_2^m
      u_m(\widehat x^m_3)  \right) \\
    & < \frac12\left(\sum_{m=1}^M w_1^m u_m(\widehat x^m_1) +
      \sum_{m=1}^M w_2^m u_m(\widehat x^m_2) \right) =
    \frac12\left(r(x,w_1) + r(x,w_2) \right).
  \end{align*}
  This finishes the verification of \ref{item:F2}.

  As we have already shown, $r=r(v,x)$ is a saddle function with
  well-defined partial derivatives at every point. In this case, $r$
  is continuously differentiable, see Theorem 35.8 and Corollary
  35.7.1 in \cite{Rock:70}, and, hence, satisfies \ref{item:F1}.

  With \ref{item:F3} following trivially from \ref{item:F4}, to
  complete the proof, we only have to verify \ref{item:F6}.  Assume
  $M>1$ and let $(w_n)_{n\geq 1}$ be a sequence in $\mathbf{S}^M$
  converging to $w\in \boundary{\mathbf{S}^M}$. For $n\geq 1$ denote
  by $\widehat x_n\in \mathbf{R}^M$ the maximal allocation of $x$
  corresponding to $w_n$.  In view of \eqref{eq:73}, $\lim_{n\to
    \infty}\widehat x^k_n = \infty$ for every index $k$ with $w^k>0$.
  As $\sum_{m=1}^M \widehat{x}^m_n = x$, there is an index $m_0$ such
  that $\lim_{n\to \infty}\widehat x^{m_0}_n = -\infty$ and,
  therefore, accounting for \eqref{eq:63} and \eqref{eq:2},
  \begin{displaymath}
    \lim_{n\to \infty} \sum_{m=1}^M \frac{\partial r}{\partial v^m}(w^m_n,x)
    \leq \lim_{n\to \infty} \frac{\partial r}{\partial v^{m_0}}(w^{m_0}_n,x)
    = \lim_{n\to \infty} u_{m_0}(\widehat x^{m_0}_n) = -\infty.
  \end{displaymath}
\end{proof}

\begin{proof}[Proof of Theorem~\ref{th:8}.] The proof relies on the
  Implicit Function Theorem. Define the function
  $\map{h=h(v,x,y,z)}{(0,\infty)^M\times \mathbf{R} \times
    \mathbf{R}^M \times \mathbf{R}}{\mathbf{R}^{M+1}}$ by
  \begin{align*}
    h^m(v,x,y,z) &= z - v^m u'_m(y^m), \quad m=1,\dots,M, \\
    h^{M+1}(v,x,y,z) &= \sum_{m=1}^M y^m - x,
  \end{align*}
  and observe that, by Theorem~\ref{th:7},
  \begin{displaymath}
    h\left(v,x,\widehat x(v,x),\frac{\partial r}{\partial x}(v,x)\right) = 0 ,
    \quad (v,x) \in (0,\infty)^M\times \mathbf{R}. 
  \end{displaymath}

  Fix $(v_0,x_0)$, set $y_0 \set \widehat x(v_0,x_0)$, $z_0 \set
  \frac{\partial r}{\partial x}(v_0,x_0)$, and denote by $B =
  (B^{kl})_{k,l = 1,\dots,M+1}$ the Jacobian of
  $h(v_0,x_0,\cdot,\cdot)$ evaluated at $(y_0,z_0)$.  Accounting for
  the fact that
  \begin{displaymath}
    v^m_0 u'_m(y^m_0) = z_0, \quad m=1,\dots,M,
  \end{displaymath}
  we deduce
  \begin{align*}
    B^{kl} &= B^{lk} = -v^k_0 u''_k(y^k_0) \delta_{kl} =
    \frac{z_0}{t_k(y^k_0)} \delta_{kl}, \quad
    k,l=1,\dots,M, \\
    B^{(M+1)m} &= B^{m(M+1)} = 1, \quad m=1,\dots,M, \\
    B^{(M+1)(M+1)} &= 0.
  \end{align*}

  Direct computations show that the inverse matrix $C \set B^{-1}$ 
  is given by
  \begin{align*}
    C^{kl} & = C^{lk} = \frac{t_k(y^k_0)}{z_0} \left(\delta_{kl} -
      \frac{t_l(y^l_0)}{\sum_{i=1}^M t_i(y^i_0)} \right), \quad
    k,l=1,\dots,M,
    \\
    C^{(M+1)m } &= C^{m(M+1)} = \frac{t_m(y^m_0)}{\sum_{i=1}^M
      t_i(y^i_0)},
    \quad m=1,\dots,M, \\
    C^{(M+1)(M+1)} &= -\frac{z_0}{\sum_{i=1}^M t_i(y^i_0)}.
  \end{align*}
  Since, for $m=1,\dots,M+1$ and $l=1,\dots,M$,
  \begin{align*}
    \frac{\partial h^m}{\partial x}(v_0,x_0,y_0,z_0) &= - \delta_{m(M+1)}, \\
    v^l \frac{\partial h^m}{\partial v^l}(v_0,x_0,y_0,z_0) &= - v^m
    u'_m(y^m_0) \delta_{lm} = - z_0 \delta_{lm},
  \end{align*}
  the Implicit Function Theorem implies the continuous
  differentiability of the functions $\widehat x = \widehat x(v,x)$
  and $\frac{\partial r}{\partial x}=\frac{\partial r}{\partial
    x}(v,x)$ in the neighborhood of $(v_0,x_0)$ and the identities:
  \begin{align*}
    \frac{\partial \widehat x^m}{\partial x}(v_0,x_0) &=
    -\sum_{k=1}^{M+1}
    C^{mk} \frac{\partial h^k}{\partial x}(v_0,x_0,y_0,z_0) = C^{m(M+1)}, \\
    v^l\frac{\partial \widehat x^m}{\partial v^l}(v_0,x_0) &=
    -\sum_{k=1}^{M+1} C^{mk} v^l \frac{\partial h^k}{\partial
      v^l}(v_0,x_0,y_0,z_0) = z_0 C^{ml},\\
    \frac{\partial^2 r}{\partial x^2}(v_0,x_0) &= -\sum_{k=1}^{M+1}
    C^{(M+1)k} \frac{\partial h^k}{\partial x}(v_0,x_0,y_0,z_0) =
    C^{(M+1)(M+1)}, \\
    v^l\frac{\partial^2 r}{\partial x\partial v^l}(v_0,x_0) &=
    -\sum_{k=1}^{M+1} C^{(M+1)k} v^l \frac{\partial h^k}{\partial
      v^l}(v_0,x_0,y_0,z_0) = z_0 C^{(M+1)l},
  \end{align*}
  proving \eqref{eq:65}--\eqref{eq:66} and
  \eqref{eq:67}--\eqref{eq:68}.

  The continuous differentiability of $\frac{\partial r}{\partial v} =
  \frac{\partial r}{\partial v}(v,x)$ with respect to $v$ and the
  identity \eqref{eq:69} follow from \eqref{eq:63} and
  \eqref{eq:66}. Direct computations relying on \eqref{eq:67},
  \eqref{eq:68}, and \eqref{eq:69} lead to the expression
  \eqref{eq:70} for $A(r)$, which jointly with \eqref{eq:67} implies
  the validity of \ref{item:F5} for $r=r(v,x)$.

  Finally, accounting for~\eqref{eq:4} and observing that~\eqref{eq:3}
  can be equivalently stated as
  \begin{displaymath}
    \frac1c \leq t_m(x) \leq c, \quad x\in \mathbf{R}, \; m=1,\dots,M,
  \end{displaymath} 
  we deduce that, for the function $r=r(v,x)$, the property
  \ref{item:F7} follows from~\eqref{eq:62} and~\eqref{eq:63},
  \ref{item:F8} is implied by~\eqref{eq:70}, and \ref{item:F9} follows
  from~\eqref{eq:67} and~\eqref{eq:68}.
\end{proof}

\section{Stochastic field of aggregate utilities and its conjugate}
\label{sec:proc-indir-util}

We remind the reader of some terminology. For a set $A\subset
\mathbf{R}^d$ a map $\map{\xi}{A}{\mathbf{L}^0(\mathbf{R}^n)}$ is
called a \emph{random field}; $\xi$ is continuous, convex, etc., if
its \emph{sample paths} $\map{\xi(\omega)}{A}{\mathbf{R}^n}$ are
continuous, convex, etc., for all $\omega\in \Omega$.  If $\xi$ and
$\eta$ are random fields on $A$ then $\eta$ is a \emph{modification}
of $\xi$ if $\xi(x) = \eta(x)$ for every $x\in A$.  A random field
$\map{X}{A\times [0,T]}{\mathbf{L}^0(\mathbf{R}^n)}$ is called a
\emph{stochastic field} if, for $t\in [0,T]$, $\map{X_t \set
  X(\cdot,t)}{A}{\mathbf{L}^0(\mathcal{F}_t,\mathbf{R}^n)}$, that is,
the random variable $X_t$ is $\mathcal{F}_t$-measurable.

Recall that the stochastic field of aggregate utilities is given by
\begin{align*}
  F_t(a) &\set \mathbb{E}[r(v,\Sigma(x,q))|\mathcal{F}_t], \quad
  a=(v,x,q)\in \mathbf{A},\; t\in [0,T],
\end{align*}
and its saddle conjugate with respect to $(v,x)$ is defined as
\begin{equation}
  \label{eq:74}
  G_t(b) \set \sup_{v\in (0,\infty)^M}\inf_{x\in
    \mathbf{R}}[\ip{v}{u} + xy - F_t(v,x,q)], \quad b=(u,y,q) \in
  \mathbf{B}.
\end{equation}
The sample paths of these stochastic fields are described in
Theorems~\ref{th:9} and~\ref{th:10} which constitute the main results
of the paper.

By $\mathbf{D}([0,T], \mathbf{X})$ we denote the space of RCLL
(right-continuous with left limits) maps of $[0,T]$ into a metric
space $\mathbf{X}$. Hereafter, for $i=1,2$, we view $\mathbf{F}^i$ as
topological subspaces of the Fr\'echet spaces
$\mathbf{C}^i(\mathbf{A})$ defined at the beginning of
Section~\ref{sec:stab-under-conv}.  A similar convention is also used
for $\widetilde{\mathbf{F}}^1$, $\widetilde{\mathbf{G}}^1$ and
$\widetilde{\mathbf{F}}^2(c)$, $\widetilde{\mathbf{G}}^2(c)$.

\begin{Theorem}
  \label{th:9}
  Suppose Assumption~\ref{as:1} and condition~\eqref{eq:5} hold. Then
  the stochastic fields $F= F_t(a)$ and $G=G_t(b)$ have modifications
  with sample paths in $\mathbf{D}([0,T], \widetilde{\mathbf{F}}^1)$
  and $\mathbf{D}([0,T], \widetilde{\mathbf{G}}^1)$, respectively, and
  their left-limits $F_{t-}(\cdot)$ and $G_{t-}(\cdot)$ are conjugate
  to each other as in~\eqref{eq:74}.  Moreover, for every compact set
  $C\subset \mathbf{A}$
  \begin{equation}
    \label{eq:75}
    \mathbb{E}[\norm{F_T(\cdot)}_{1,C}]<\infty, 
  \end{equation}
  and, for $a=(v,x,q) \in \mathbf{A}$, $t\in [0,T]$, and
  $i=1,\dots,M+1+J$,
  \begin{equation}
    \label{eq:76}
    \frac{\partial F_t}{\partial a^i}(a) = 
    \mathbb{E}[\frac{\partial F_T}{\partial a^i}(a)|\mathcal{F}_t]. 
  \end{equation}
\end{Theorem}

\begin{Theorem}
  \label{th:10}
  Suppose Assumptions~\ref{as:1} and~\ref{as:2} and the
  condition~\eqref{eq:5} hold. Then the stochastic fields $F= F_t(a)$
  and $G=G_t(b)$ have modifications with sample paths in
  $\mathbf{D}([0,T],\widetilde{\mathbf{F}}^2(c))$ and
  $\mathbf{D}([0,T],\widetilde{\mathbf{G}}^2(c))$, respectively, with
  the constant $c>0$ from Assumption~\ref{as:2}. Moreover, for every
  compact set $C\subset \mathbf{A}$
  \begin{equation}
    \label{eq:77}
    \mathbb{E}[\norm{F_T(\cdot)}_{2,C}]<\infty, 
  \end{equation}
  and, for $a=(v,x,q) \in \mathbf{A}$, $t\in [0,T]$, and
  $i,j=1,\dots,M+1+J$,
  \begin{equation}
    \label{eq:78}
    \frac{\partial^2 F_t}{\partial a^i\partial a^j}(a) =
    \mathbb{E}[\frac{\partial^2 F_T}{\partial a^i \partial
      a^j}(a)|\mathcal{F}_t].  
  \end{equation}
\end{Theorem}

Theorems~\ref{th:1}, \ref{th:2}, and \ref{th:5} in
Section~\ref{sec:spac-saddle-funct} allow us to establish various
identities between the first and second derivatives of $F_t(\cdot)$
and $G_t(\cdot)$. We state one such corollary mentioned in
Section~\ref{sec:setup}. Recall that $\mathbf{C}^0$ denotes the space
of continuous functions with the topology of uniform convergence on
compacts sets; see Section~\ref{sec:stab-under-conv}.

\begin{Corollary}
  \label{cor:1}
  Suppose Assumption~\ref{as:1} and the condition~\eqref{eq:5}
  hold. Then the stochastic fields
  \begin{displaymath}
    U^m_t(a) \set \frac{\partial F_t}{\partial v^m}(a), \quad
    a=(v,x,q)\in
    \mathbf{A}, \; m=1,\dots,M, 
  \end{displaymath}
  have sample paths in $\mathbf{D}([0,T],\mathbf{C}^0(\mathbf{A}))$,
  the stochastic fields
  \begin{align*}
    X_t(u,q) &\set G_t(b), \quad
    b=(u,1,q)\in \mathbf{B},\\
    V^m_t(u,q) &\set \frac{\partial G_t}{\partial
      u^m}\frac{1}{\sum_{l=1}^M \frac{\partial G_t}{\partial u^l}}(b),
    \quad b=(u,1,q)\in \mathbf{B}, \; m=1,\dots,M,
  \end{align*}
  have sample paths in $\mathbf{D}([0,T],\mathbf{C}^0((-\infty,0)^M,
  \mathbf{R}^J))$, and the following invertibility relations hold:
  \begin{align*}
    u^m &= U^m_t(V_t(u,q), X_t(u,q),q), \; m=1,\dots,M, \\
    x &= X_t(U_t(v,x,q),1,q), \\
    v^m &= V^m_t(U_t(v,x,q),1,q), \; m=1,\dots,M,
  \end{align*}
  where $u\in (-\infty,0)^M$, $x\in \mathbf{R}$, $v\in \mathbf{S}^M$,
  $q\in \mathbf{R}^J$, and $t\in [0,1]$.

  If, in addition, Assumption~\ref{as:2} holds then these stochastic
  fields have sample paths in $\mathbf{D}([0,T],\mathbf{C}^1)$.
\end{Corollary}
\begin{proof}
  The result follows directly from Theorems~\ref{th:9} and \ref{th:10}
  and the conjugacy relations in items~\ref{item:6} and \ref{item:7}
  of Theorem~\ref{th:1} as soon as we account for the positive
  homogeneity property~\eqref{eq:20} of the elements of
  $\mathbf{G}^1$.
\end{proof}

\begin{Remark}
  \label{rem:3}
  Consider the price impact model from Section~\ref{sec:setup}.
  Recall the definition of Pareto allocation $\pi(a) = (\pi^m(a))$
  from~\eqref{eq:8} and observe that by Theorems~\ref{th:7}
  and~\ref{th:9}
  \begin{displaymath}
    U^m_t(a) = \frac{\partial F_t}{\partial v^m}(a) =
    \mathbb{E}[u_m(\pi^m(a))|\mathcal{F}_t].
  \end{displaymath}
  Hence, $U^m_t(a)$ represents the \emph{expected utility} of the
  $m$th market maker at time $t$ given the Pareto allocation
  $\pi(a)$. By the invertibility relations in Corollary~\ref{cor:1},
  the random variables $X_t(u,q)$ and $V_t(u,q)$ define the
  \emph{collective cash amount} and the \emph{Pareto weights} of the
  market makers at time $t$ when their current expected utilities are
  given by $u$ and they jointly own $q$ stocks.
\end{Remark}

\subsection{Proofs of Theorems \ref{th:9} and \ref{th:10}}
\label{sec:proof-theorem}

For the proof of Theorem~\ref{th:9} we need the following result
linking the condition \ref{item:F6} in the definition of
$\widetilde{\mathbf{F}}^1$ with the condition \ref{item:23} in
Lemma~\ref{lem:27} from Appendix~\ref{sec:modif-rand-fields}.

\begin{Lemma}
  \label{lem:16}
  Let $M>1$. A function $f\in \mathbf{F}^1$ satisfies {\rm
    \ref{item:F6}} (that is, belongs to $\widetilde{\mathbf{F}}^1$) if
  and only if for every increasing sequence $(C_n)_{n\geq 1}$ of
  compact sets in $\mathbf{S}^M$ with $\cup_{n\geq 1} C_n =
  \mathbf{S}^M$ and for every compact set $D\subset \mathbf{R}^{1+J}$
  \begin{equation}
    \label{eq:79}
    \lim_{n\to\infty}\sup_{w\in \mathbf{S}^M/C_n}\sup_{(x,q)\in D} \sum_{m=1}^M
    \frac{\partial f}{\partial v^m}(w,x,q) 
    = -\infty.
  \end{equation}
\end{Lemma}
\begin{proof} The ``if'' statement is straightforward. Hereafter we
  shall focus on the opposite implication.

  To verify~\eqref{eq:79} we have to show that for $f\in
  \widetilde{\mathbf{F}}^1$ and every $a_n = (w_n,x_n,q_n)\in
  \mathbf{S}^M\times \mathbf{R}\times \mathbf{R}^J$, $n\geq 1$,
  converging to $(w,x,q) \in \partial \mathbf{S}^M\times
  \mathbf{R}\times \mathbf{R}^J$ we have
  \begin{equation}
    \label{eq:80}
    \lim_{n\to\infty} \sum_{m=1}^M 
    \frac{\partial f}{\partial v^m}(a_n) = \lim_{n\to\infty}
    \ip{\frac{\partial f}{\partial v}(a_n)}{\idvec} = -\infty,
  \end{equation}
  where $\idvec \set (1,\dots,1)$.

  Let $\epsilon>0$. Accounting for the convexity and the positive
  homogeneity of the functions $f(\cdot,x_n,q_n)$, $n\geq 1$, on
  $(0,\infty)^M$ we deduce 
  \begin{gather*}
    \lim_{n\to\infty} \ip{\frac{\partial f}{\partial v}(a_n)}{\idvec}
    \leq \lim_{n\to\infty}
    \ip{\frac{\partial f}{\partial v}(w_n+\epsilon\idvec,x_n,q_n)}{\idvec} \\
    = \ip{\frac{\partial f}{\partial v}(w+\epsilon\idvec,x,q)}{\idvec}
    = \ip{\frac{\partial f}{\partial v}(w(\epsilon),x,q)}{\idvec},
  \end{gather*}
  where $w(\epsilon) \set \frac{w+\epsilon\idvec}{1+\epsilon M}$
  belongs to $\mathbf{S}^M$.  By \ref{item:F6}, the passage to the
  limit when $\epsilon \to 0$ yields \eqref{eq:80}.
\end{proof}

\begin{proof}[Proof of Theorem~\ref{th:9}.]
  As
  \begin{displaymath}
    F_T(a) = r(v,\Sigma(x,q)), \quad a=(v,x,q)\in \mathbf{A}, 
  \end{displaymath}
  the assertions concerning the sample paths of $F_T=F_T(a)$ are
  immediate corollaries of the corresponding properties of $r=r(v,x)$
  established in Theorem~\ref{th:7}. By~\eqref{eq:5} we have that
  $F_T(a)\in \mathbf{L}^1$, $a\in \mathbf{A}$, and then, by
  Theorem~\ref{th:12} from Appendix~\ref{sec:unif-integr-saddle},
  obtain \eqref{eq:75}. Lemma~\ref{lem:25} then implies that $F$ has
  sample paths in $\mathbf{D}([0,T], \mathbf{C}^1(\mathbf{A}))$ and
  that the equality~\eqref{eq:76} holds.

  To verify that the sample paths of $F$ belong to
  $\mathbf{D}([0,T],\mathbf{F}^1)$ it is sufficient to match the
  properties \ref{item:F1}--\ref{item:F4} in the description of
  $\mathbf{F}^1$ with the properties \ref{item:15}--\ref{item:22} in
  Lemmas~\ref{lem:26} and~\ref{lem:27}. For the most part these
  correspondences are straightforward with the links between
  \eqref{eq:12} in \ref{item:F2} or \eqref{eq:13} in \ref{item:F4} and
  their respective versions of \ref{item:22} holding due to the
  equivalence of the pointwise and the uniform on compact sets
  convergences for a sequence of convex or saddle functions.

  Note that in order to use \ref{item:22} in Lemma~\ref{lem:27} we
  still have to verify the integrability condition~\eqref{eq:104} in
  that lemma. Its adaption to \eqref{eq:12} in \ref{item:F2} takes the
  form:
  \begin{displaymath}
    \mathbb{E}[\inf_{w\in \mathbf{S}^M} \inf_{(x,q)\in D} F_T(w,x,q)]
    > -\infty, 
  \end{displaymath}
  for every compact set $D\subset \mathbf{R}^{1+J}$, and follows
  because of \eqref{eq:75} and since the sample paths of $F_T(\cdot)$
  are decreasing with respect to $v$ (hence, $F_T(w,x,q) >
  F_T(\idvec,x,q)$, $w\in \mathbf{S}^M$, where $\idvec \set
  (1,\dots,1)$).  To perform a similar verification for the
  convergence~\eqref{eq:13} in \ref{item:F4} we restrict the domain of
  $x$ to $[0, \infty)$. The analog of~\eqref{eq:104} then has the
  form:
  \begin{displaymath}
    \mathbb{E}[\inf_{x\geq 0} \inf_{(v,q)\in D} F_T(v,x,q)] > -\infty,
  \end{displaymath}
  for every compact set $D\subset (0,\infty)^M\times \mathbf{R}^J$,
  and follows from~\eqref{eq:75} and the monotonicity of $F$ with
  respect to $x$. Thus, we have shown that the sample paths of $F$
  belong to $\mathbf{D}([0,T],\mathbf{F}^1)$.

  The connection between \ref{item:23} and \ref{item:F6} has been
  established in Lemma~\ref{lem:16}. The adaptation of \eqref{eq:104}
  to this case holds trivially as $\frac{\partial F}{\partial v} <
  0$. Hence, the sample paths of $F$ belong to
  $\mathbf{D}([0,T],\widetilde{\mathbf{F}}^1)$.

  This finishes the proof of the assertions related to $F$.  The
  remaining part, concerning $G$, follows directly from
  Theorems~\ref{th:6} and~\ref{th:3}.
\end{proof}

We divide the proof of Theorem \ref{th:10} into lemmas.

\begin{Lemma}
  \label{lem:17}
  Under the conditions of Theorem~\ref{th:10} the sample paths of the
  random field $F_T=F_T(\cdot)$ belong to the space
  $\widetilde{\mathbf{F}}^2(c)$ with the same constant $c>0$ as in
  Assumption~\ref{as:2}. Moreover, \eqref{eq:77} holds for every
  compact set $C\subset \mathbf{A}$.
\end{Lemma}

\begin{proof} The assertions for the sample paths of
  \begin{displaymath}
    F_T(a) = r(v,\Sigma(x,q)), \quad a=(v,x,q)\in \mathbf{A}, 
  \end{displaymath}
  follow directly from the properties of $r=r(v,x)$ in
  Theorem~\ref{th:8}.

  To verify \eqref{eq:77} fix a compact set $C\subset \mathbf{A}$.
  From the formulas for the second derivatives of $r=r(v,x)$ in
  Theorem~\ref{th:8} and from Assumption~\ref{as:2} we deduce the
  existence of a constant $b_1>0$ such that
  \begin{equation}
    \label{eq:81}
    \abs{\frac{\partial^2 F_T}{\partial a^i\partial a^j}(a)} \leq b_1
    \frac{\partial F_T}{\partial x}(a) (1+\abs{\psi}^2), \quad a\in
    C. 
  \end{equation}

  From~\eqref{eq:67} we deduce that
  \begin{displaymath}
    \frac1c \frac{\partial r}{\partial x}(v,x) \leq -M\frac{\partial^2
      r}{\partial x^2}(v,x) \leq c \frac{\partial r}{\partial x}(v,x),
    \quad x \in \mathbf{R},
  \end{displaymath}
  where $c>0$ is the constant from Assumption~\ref{as:2}. This yields
  the exponential growth property
  \begin{displaymath}
    e^{-y^+c/M+y^-/(cM)} \leq \frac{\frac{\partial r}{\partial
        x}(v,x+y)}{\frac{\partial
        r}{\partial x}(v,x)} \leq e^{-y^+/(cM)+y^-c/M}, \quad x,y \in \mathbf{R}, 
  \end{displaymath}
  where $x^+\set \max(x,0)$ and $x^- \set (-x)^+$.

  It follows that the right side of~\eqref{eq:81} is dominated by
  $b_2\norm{F_T}_{1,D}$ for some constant $b_2>0$ and a compact set
  $D$ in $\mathbf{A}$ containing $C$. Hence, there is a constant
  $b_3>0$ such that
  \begin{displaymath}
    \norm{F_T}_{2,C} \leq b_3 \norm{F_T}_{1,D}
  \end{displaymath}
  and \eqref{eq:77} holds because $\norm{F_T}_{1,D}$ has a finite
  expected value by Theorem~\ref{th:9}.
\end{proof}

\begin{Lemma}
  \label{lem:18}
  Under the conditions of Theorem~\ref{th:10} the stochastic field $F
  = F_t(a)$ has sample paths in
  $\mathbf{D}([0,T],\mathbf{C}^2(\mathbf{A}))$ and~\eqref{eq:78}
  holds.
\end{Lemma}
\begin{proof}
  By Lemma~\ref{lem:17}, the random field $F_T(\cdot)$ has sample
  paths in $\mathbf{C}^2(\mathbf{A})$ and~\eqref{eq:77} holds. The
  result then follows from Lemma~\ref{lem:25}.
\end{proof}

Recall the notation $A(f)$ for the matrix defined in \eqref{eq:14}. A
delicate point in the proof of Theorem~\ref{th:10} is to verify
\ref{item:F8} for the matrices $A(F_t)$, $t\in [0,T]$.

For $a\in \mathbf{A}$, define the probability measure $\mathbb{R}(a)$
with
\begin{equation}
  \label{eq:82}
  \frac{d\mathbb{R}(a)}{d\mathbb{P}} \set \frac{\partial^2 F_T}{\partial
    x^2}(a)/\frac{\partial^2 F_0}{\partial x^2}(a), 
\end{equation}
the stochastic process
\begin{equation}
  \label{eq:83}
  R_t(a) \set -\frac{\partial F_t}{\partial
    x}(a)/\frac{\partial^2 F_t}{\partial x^2}(a),\quad t\in [0,T], 
\end{equation}
and the random variables
\begin{equation}
  \label{eq:84}
  \tau^m(a) \set t_m(\pi^m(a)), \quad m=1,\dots,M,
\end{equation}
where $t_m = t_m(x)$ is the absolute risk-tolerance of $u_m =u_m(x)$
and $\pi^m(a)$ is a Pareto optimal allocation; see~\eqref{eq:64}
and~\eqref{eq:8}. Observe that $R(a)$ is a martingale under
$\mathbb{R}(a)$, that, by \eqref{eq:67},
\begin{equation}
  \label{eq:85}
  \sum_{m=1}^M \tau^m(a) = {R_T(a)},
\end{equation}
and, that by Assumption~\ref{as:2},
\begin{equation}
  \label{eq:86}
  \frac1c \leq \tau^m\leq c, \quad m=1,\dots,M.
\end{equation}

\begin{Lemma}
  \label{lem:19}
  Assume that the conditions of Theorem~\ref{th:10} hold. Then the
  matrix $A(F_t)(a)$ is given by
  \begin{displaymath}
    \begin{split}
      A^{lm}(F_t)(a) &= \frac1{R_t(a)}
      \mathbb{E}_{\mathbb{R}(a)}[\tau^l(a)(\delta_{lm}\sum_{k=1}^M
      \tau^k(a) - \tau^m(a))|\mathcal{F}_t] \\
      & +
      \frac1{R_t(a)}\mathbb{E}_{\mathbb{R}(a)}[\tau^l(a)|\mathcal{F}_t]
      \mathbb{E}_{\mathbb{R}(a)}[\tau^m(a)|\mathcal{F}_t], \quad
      l,m=1,\dots,M,
    \end{split}
  \end{displaymath}
  where the probability measure $\mathbb{R}(a)$, the stochastic
  process $R(a)$, and the random variable $\tau(a)$ are defined in
  \eqref{eq:82}, \eqref{eq:83}, and \eqref{eq:84}, respectively and
  $\delta_{lm} \set\ind{l=m}$ is the Kronecker delta.

  Moreover, for every $z\in \mathbf{R}^n$,
  \begin{equation}
    \label{eq:87} 
    \frac1{c} \abs{z}^2 \leq \ip{z}{A(F_t)(a)z} \leq  c\abs{z}^2, 
  \end{equation}
  where the constant $c>0$ is given in Assumption~\ref{as:2}.
\end{Lemma}

\begin{proof}
  From the expressions for the second derivatives of $r=r(v,x)$ in
  Theorem~\ref{th:8} we deduce
  \begin{align*}
    v^m\frac{\partial^2 F_T}{\partial v^m\partial x}(a) &=
    -\frac{\partial^2
      F_T}{\partial x^2}(a) \tau^m(a), \\
    v^lv^m\frac{\partial^2 F_T}{\partial v^l\partial v^m}(a) &=
    -\frac{\partial^2 F_T}{\partial x^2}(a) \tau^l(a)(\delta_{lm}
    \sum_{k=1}^M \tau^k(a) - \tau^m(a)),
  \end{align*}
  and the expression for $A(F_t)$ follows by direct computations.

  To simplify notations, in the proof of~\eqref{eq:87}, we omit the
  dependence on $a\in \mathbf{A}$ and consider the case
  $t=0$. Elementary calculations show that
  \begin{align*}
    \ip{z}{A(F_0)z} &=
    \frac1{R_0}\left(\mathbb{E}_{\mathbb{R}}\left[R_T\sum_{m=1}^M
        \tau^m z_m^2 - \ip{\tau}{z}^2\right] +
      \ip{\mathbb{E}_{\mathbb{R}}[\tau]}{z}^2\right) \\
    &= \frac1{R_0}\mathbb{E}_{\mathbb{R}}\left[R_T\sum_{m=1}^M \tau^m
      z_m^2 - \ip{\tau - \mathbb{E}_{\mathbb{R}}[\tau]}{z}^2\right],
  \end{align*}
  where we used~\eqref{eq:85}.  This immediately implies the upper
  bound in~\eqref{eq:87}:
  \begin{displaymath}
    \ip{z}{A(F_0)z}  \leq \frac1{R_0}
    \mathbb{E}_{\mathbb{R}}[R_T\sum_{m=1}^M \tau^m z_m^2] \leq c
    \abs{z}^2,  
  \end{displaymath}
  where we used the inequality~\eqref{eq:86} and the martingale
  property of $R$ under $\mathbb{R}$.

  To verify the lower bound observe that by \eqref{eq:86}
  \begin{displaymath}
    \theta^m \set \tau^m - \frac1c \geq 0.
  \end{displaymath}
  We obtain
  \begin{align*}
    \ip{z}{A(F_0)z} &= \frac1{R_0} \mathbb{E}_{\mathbb{R}}\left[R_T
      \sum_{m=1}^M (\frac1c + \theta^m) z_m^2 -
      \ip{\theta - \mathbb{E}_{\mathbb{R}}[\theta]}{z}^2 \right] \\
    & = \frac1c \abs{z}^2 + \frac1{R_0}
    \mathbb{E}_{\mathbb{R}}\left[R_T \sum_{m=1}^M \theta^m z_m^2 -
      \ip{\theta - \mathbb{E}_{\mathbb{R}}[\theta]}{z}^2 \right].
  \end{align*}
  As $R_T = \ip{\tau}{\idvec} \geq \ip{\theta}{\idvec}$, where $\idvec
  \set (1,\dots,1)$, we deduce
  \begin{gather*}
    R_0\left(\ip{z}{A(F_0)z} - \frac1c \abs{z}^2\right) \geq
    \mathbb{E}_{\mathbb{R}}\left[\ip{\theta}{\idvec} \sum_{m=1}^M
      \theta^m
      z_m^2 - \ip{\theta - \mathbb{E}_{\mathbb{R}}[\theta]}{z}^2 \right] \\
    = \mathbb{E}_{\mathbb{R}}\left[\frac1{\ip{\theta}{\idvec}}
      \sum_{m=1}^M \theta^m(z_m\ip{\theta}{\idvec} - \ip{\theta}{z})^2
    \right] + \ip{\mathbb{E}_{\mathbb{R}}[\theta]}{z}^2 \geq 0.
  \end{gather*}
\end{proof}

\begin{proof}[Proof of Theorem~\ref{th:10}.] The
  inequality~\eqref{eq:77} and the fact that the sample paths of
  $F_T(\cdot)$ belong to $\widetilde{\mathbf{F}}^2(c)$ are established
  in Lemmas~\ref{lem:17}, while the identity~\eqref{eq:78} and the
  fact that the sample paths of $F=F_t(a)$ belong to
  $\mathbf{D}([0,T],\mathbf{C}^2(\mathbf{A}))$ are proved in
  Lemma~\ref{lem:18}. The remaining properties of $F$ follow from
  Lemma~\ref{lem:26} if we account for the properties of the sample
  paths for $F_T(\cdot)$ and use Lemma~\ref{lem:19}. Finally, the
  properties of the sample paths for $G=G_t(b)$ follow directly from
  those of $F$ and Theorems~\ref{th:6} and~\ref{th:4}.
\end{proof}

\appendix
\section{An envelope theorem for saddle functions}
\label{sec:envelope-theorem}

In the proof of Theorem~\ref{th:1} we used the following version of
the folklore ``envelope'' theorem for saddle functions. As usual, $\ri
C$ denotes the \emph{relative interior} of a convex set $C$.

\begin{Theorem}
  \label{th:11}
  Let $C$ be a convex set in $\mathbf{R}^n$, $D$ be a convex set in
  $\mathbf{R}^m$, $E$ be a convex open set in $\mathbf{R}^l$, $f =
  f(x,y,z): C\times D \times E \rightarrow \mathbf{R}$ be a function
  convex with respect to $x$ and concave with respect to $(y,z)$, and
  let $z_0\in E$. Denote
  \begin{displaymath}
    g(z) \set \sup_{y\in D}\inf_{x\in C} f(x,y,z), \quad z\in E, 
  \end{displaymath}
  and assume that the maximin value $g(z_0)$ is attained at a unique
  $x_0\in \ri C$ and some (not necessarily unique) $y_0\in D$ and the
  function $f(x_0,y_0,\cdot)$ is differentiable at $z_0$.

  Then the function $\map{g}{E}{\mathbf{R}\cup \braces{-\infty}}$ is
  concave, differentiable at $z_0$ (in particular, finite in a
  neighborhood of $z_0$) and
  \begin{displaymath}
    \nabla g(z_0) = \frac{\partial f}{\partial z}(x_0,y_0,z_0). 
  \end{displaymath}
\end{Theorem}

\begin{Remark}
  \label{rem:4}
  Theorem 5 in \citet{MilgSegal:02} is the closest result to ours we
  could find in the literature. There, the convexity assumptions on
  $f$ are replaced by compactness requirements on $C$ and $D$.
\end{Remark}

The proof of Theorem~\ref{th:11} relies on two lemmas of independent
interest, which were used in the proof of Theorem~\ref{th:7}.  The
first lemma is essentially known, see, for example, Corollary 3 in
\cite{MilgSegal:02}.

\begin{Lemma}
  \label{lem:20}
  Let $f = f(x,y): \mathbf{R}^n\times \mathbf{R}^m \rightarrow
  \mathbf{R}\cup \braces{-\infty}$ be a concave function and let
  $y_0\in \mathbf{R}^m$. Denote
  \begin{displaymath}
    g(y) \set \sup_{x\in \mathbf{R}^n} f(x,y), \quad y\in \mathbf{R}^m, 
  \end{displaymath}
  and assume that the upper bound $g(y_0)$ is attained at some (not
  necessarily unique) $x_0\in \mathbf{R}^n$ and the function
  $f(x_0,\cdot)$ is differentiable at $y_0$.

  Then the function $\map{g}{\mathbf{R}^m}{\mathbf{R}\cup
    \braces{-\infty}}$ is concave, differentiable at $y_0$ and
  \begin{equation}
    \label{eq:88}
    \nabla g(y_0) = \frac{\partial f}{\partial
      y}(x_0,y_0).
  \end{equation}
\end{Lemma}

\begin{proof} The concavity of $g$ follows from the concavity of $f$
  with respect to both arguments. As $g(y_0)=f(x_0,y_0)<\infty$, this
  concavity property implies that $g<\infty$. Since $g\geq
  f(x_0,\cdot)$, the function $g$ is finite in a neighborhood of
  $y_0$. It follows that $\partial g(y_0)$, the subdifferential of $g$
  at $y_0$, is not empty.

  If $y^*\in \partial g(y_0)$, then
  \begin{displaymath}
    g(y) \leq g(y_0) + \ip{y^*}{y-y_0}, \quad y\in \mathbf{R}^m.
  \end{displaymath}
  As $f(x_0,y) \leq g(y)$, $y\in \mathbf{R}^m$, and $f(x_0,y_0) =
  g(y_0)$, it follows that
  \begin{displaymath}
    f(x_0,y) \leq  f(x_0,y_0) + \ip{y^*}{y-y_0}, \quad y\in \mathbf{R}^m.
  \end{displaymath}
  Hence, $y^*$ belongs to the subdifferential of $f(x_0,\cdot)$ at
  $y_0$, and, therefore, $y^* = \frac{\partial f}{\partial
    y}(x_0,y_0)$.  It follows that $y^*$ is the only element of
  $\partial g(y_0)$, proving the differentiability of $g$ at $y_0$ and
  the identity \eqref{eq:88}.
\end{proof}

\begin{Lemma}
  \label{lem:21}
  Let $C$ be a convex set in $\mathbf{R}^n$, $D$ be a convex open set
  in $\mathbf{R}^m$, $f = f(x,y): C\times D \rightarrow \mathbf{R}$ be
  a function concave with respect to $x$ and convex with respect to
  $y$, and let $y_0\in D$. Define the function
  \begin{displaymath}
    g(y) \set \sup_{x\in C} f(x,y), \quad y\in D, 
  \end{displaymath}
  and assume that the upper bound $g(y_0)$ is attained at a unique
  $x_0 \in \ri C$ and the function $f(x_0,\cdot)$ is differentiable at
  $y_0$.

  Then the function $\map{g}{D}{\mathbf{R}\cup \braces{\infty}}$ is
  convex, differentiable at $y_0$, and the identity~\eqref{eq:88}
  holds.
\end{Lemma}

\begin{Remark}
  \label{rem:5}
  The proof of Lemma~\ref{lem:21} will follow from the well-known
  analogous result in convex optimization, where the assumption of
  concavity in $x$ is replaced by the requirement that $C$ is compact,
  see, for example, Corollary 4.45 in \citet{HiriarLemar:01}.
\end{Remark}

\begin{proof}
  The convexity of $g$ is straightforward.  Let $\epsilon>0$ be such
  that
  \begin{displaymath}
    C(\epsilon) \set \descr{x\in C}{\abs{x-x_0}\leq \epsilon} \subset \ri
    C. 
  \end{displaymath}
  If $(y_n)_{n\geq 1}$ is a sequence in $D$ converging to $y_0$, then
  the concave functions $f(\cdot,y_n)$, $n\geq 1$, converge to
  $f(\cdot,y_0)$ uniformly on compact subsets of $C$. Since $x_0$ is
  the unique point of maximum for $f(\cdot,y_0)$, there is $n_0>0$
  such that for every $n\geq n_0$ the concave function $f(\cdot,y_n)$
  attains its maximum at some point $x_n \in C(\epsilon)$. This
  argument implies the existence of $\delta>0$ such that
  \begin{displaymath}
    g(y) = \sup_{x\in C(\epsilon)} f(x,y)<\infty, \quad y\in D, \abs{y-y_0}<
    \delta. 
  \end{displaymath}
  As $C(\epsilon)$ is a compact set, the result now follows from the
  well-known fact in convex optimization mentioned in
  Remark~\ref{rem:5}.
\end{proof}

\begin{proof}[Proof of Theorem~\ref{th:11}] The function $h=h(y,z):
  D\times E \rightarrow \mathbf{R}\cup \braces{-\infty}$ given by
  \begin{displaymath}
    h(y,z) \set \inf_{x\in C} f(x,y,z), \quad y\in D, z\in E, 
  \end{displaymath}
  is clearly concave.  By Lemma~\ref{lem:21} the function
  $h(y_0,\cdot)$ is differentiable at $z_0$ and $\frac{\partial
    h}{\partial z}(y_0,z_0) = \frac{\partial f}{\partial z}(x_0,
  y_0,z_0)$. An application of Lemma~\ref{lem:20} completes the proof.
\end{proof}

\section{Integrability of saddle random fields}
\label{sec:unif-integr-saddle}

The following theorem states that the pointwise integrability of a
saddle random field $\xi$ implies the integrability of pseudo-norms
$\norm{\xi}_{0,C}$ for compact sets $C$. It also implies the
integrability of $\norm{\xi}_{1,C}$ if, in addition, the sample paths
of $\xi$ are differentiable. See Section~\ref{sec:stab-under-conv} for
the definition of the semi-norm $\norm{\cdot}_{m,C}$.  This result is
used in the proof of Theorem~\ref{th:9}.

We fix a probability space $(\Omega, \mathcal{F},\mathbb{P})$ and
denote $\mathbf{L}^1 = \mathbf{L}^1(\Omega, \mathcal{F},\mathbb{P})$
the Banach space of integrable random variables.

\begin{Theorem}
  \label{th:12}
  Let $U\subset \mathbf{R}^n$ and $V\subset \mathbf{R}^m$ be open sets
  and $\map{\xi}{U\times V}{\mathbf{L}^1}$ be a random field with
  sample paths in the space of concave-convex functions on $U\times
  V$.  Then for every compact set $C\subset U\times V$
  \begin{equation}
    \label{eq:89}
    \mathbb{E}[\norm{\xi}_{0,C}] < \infty. 
  \end{equation}
  If, in addition, the sample paths of $\xi$ belong to
  $\mathbf{C}^1(U)$, then
  \begin{equation}
    \label{eq:90}
    \mathbb{E}[\norm{\xi}_{1,C}] < \infty.
  \end{equation}
\end{Theorem}

The proof is divided into lemmas.

\begin{Lemma}
  \label{lem:22}
  Let $U$ be an open set in $\mathbf{R}^d$, $\map{f}{U}{\mathbf{R}}$
  be a convex function, $C$ be a compact subset of $U$, and
  $\epsilon>0$ be such that
  \begin{equation}
    \label{eq:91}
    C(\epsilon) \set \descr{x\in \mathbf{R}^d}{\inf_{y\in C}\abs{x-y} \leq
      \epsilon } \subset U. 
  \end{equation}
  Then for every $y\in C$ we have
  \begin{equation}
    \label{eq:92}
    \min_{x\in C}f(x) \geq f(y) + \frac{\sup_{x\in C}|x-y|}{\epsilon} \left(f(y) -
      \max_{x\in C(\epsilon)} f(x)\right). 
  \end{equation}
\end{Lemma}

\begin{proof} Fix $y\in C$. For every $x\in C$ there is $z\in
  \boundary{C}(\epsilon)$ such that $y$ is a convex combination of $x$
  and $z$: $y = tx + (1-t) z$ for some $t\in (0,1)$.  Using the fact
  that $|y-z| \geq \epsilon$ we obtain
  \begin{equation}
    \label{eq:93}
    \frac{1-t}{t} = \frac{\abs{x-y}}{\abs{y-z}} \leq \frac{\sup_{x\in
        C}\abs{x-y}}{\epsilon}. 
  \end{equation}
  The convexity of $f$ implies
  \begin{displaymath}
    f(y) \leq tf(x) + (1-t) f(z), 
  \end{displaymath}
  or, equivalently,
  \begin{displaymath}
    f(x) \geq f(y) + \frac{1-t}{t}(f(y) - f(z)),
  \end{displaymath}
  which, in view of \eqref{eq:93}, yields \eqref{eq:92}.
\end{proof}

\begin{Lemma}
  \label{lem:23}
  In addition to the conditions of Lemma~\ref{lem:22} assume that
  $f\in \mathbf{C}^1(U)$. Then
  \begin{equation}
    \label{eq:94}
    \norm{f}_{1,C} \leq \left(\frac{2\sqrt{d}}{\epsilon}+1\right)
    \norm{f}_{0,C(\epsilon)}.  
  \end{equation}
\end{Lemma}

\begin{proof}
  For $y\in C$ and $x\in C(\epsilon)$ we obtain from the convexity of
  $f$:
  \begin{displaymath}
    f(x) - f(y) \geq \ip{x-y}{\nabla f(y)}. 
  \end{displaymath}
  It follows that
  \begin{displaymath}
    2 \norm{f}_{0, C(\epsilon)} \geq \sup_{y\in C}
    \sup_{\descr{x}{\abs{x-y}\leq \epsilon}} \left(\ip{x-y}{\nabla
        f(y)}\right) = \epsilon \sup_{y\in C} \abs{\nabla f(y)}.   
  \end{displaymath}
  Since,
  \begin{displaymath}
    \abs{\nabla f(y)} \set \sqrt{\sum_{i=1}^d \left(\frac{\partial f}{\partial
          x_i}(y)\right)^2 } \geq \frac{1}{\sqrt{d}} \sum_{i=1}^d
    \abs{\frac{\partial f}{\partial x_i}(y)},
  \end{displaymath}
  we obtain
  \begin{displaymath}
    \abs{f(y)} +  \sum_{i=1}^d \abs{\frac{\partial f}{\partial 
        x_i}(y)} \leq \left(\frac{2\sqrt{d}}{\epsilon}+1\right)
    \norm{f}_{0,C(\epsilon)},
  \end{displaymath}
  which clearly implies \eqref{eq:94}.
\end{proof}

\begin{Lemma}
  \label{lem:24}
  Let $U$ be an open set in $\mathbf{R}^d$,
  $\map{\xi=\xi(x)}{U}{\mathbf{L}^1}$ be a random field with sample
  paths in the space of convex functions on $U$. Then, for every
  compact set $C\subset U$,
  \begin{equation}
    \label{eq:95}
    \mathbb{E}[\norm{\xi}_{0,C}] < \infty.
  \end{equation}
  If, in addition, the sample paths of $\xi$ belong to
  $\mathbf{C}^1(U)$ then
  \begin{equation}
    \label{eq:96}
    \mathbb{E}[\norm{\xi}_{1,C}] < \infty.
  \end{equation}
\end{Lemma}

\begin{proof}
  Let us first show that for every compact set $C\subset U$
  \begin{equation}
    \label{eq:97}
    \mathbb{E}[\max_{x\in C} \xi(x)] < \infty. 
  \end{equation}
  Without restricting generality we can assume that $C$ is the closed
  convex hull of a finite family $(x_i)_{i=1,\dots,I}$ in $U$. From
  the convexity of $\xi$ we then deduce
  \begin{displaymath}
    \max_{x\in C} \xi(x) = \max_{i=1,\dots,I}\xi(x_i),
  \end{displaymath}
  and \eqref{eq:97} follows from the assumption $\xi(x)\in
  \mathbf{L}^1$, $x\in U$.

  Since $C$ is a compact set in $U$, for sufficiently small $\epsilon
  >0$ the set $C(\epsilon)$ defined in \eqref{eq:91} is also a compact
  subset of $U$. By \eqref{eq:97} and Lemma~\ref{lem:22} we obtain
  \begin{displaymath}
    \mathbb{E}[\min_{x\in C} \xi(x)] > - \infty,
  \end{displaymath}
  which implies \eqref{eq:95}.

  Finally, if $f\in \mathbf{C}^1$, then \eqref{eq:96} follows from
  \eqref{eq:95} and Lemma~\ref{lem:23}.
\end{proof}

\begin{proof}[Proof of Theorem~\ref{th:12}.]
  It is sufficient to consider the case $C = C_1 \times C_2$, where
  $C_1$ and $C_2$ are compact subsets of $U$ and $V$, respectively. To
  prove~\eqref{eq:89} it is enough to show that
  \begin{equation}
    \label{eq:98}
    \sup_{x\in C_1} \sup_{y\in C_2} \xi(x,y) \in \mathbf{L}^1. 
  \end{equation}
  Indeed, having established~\eqref{eq:98} for every random field
  $\xi$ and every pair of open sets $U$ and $V$ satisfying the
  conditions of the lemma we deduce
  \begin{displaymath}
    \inf_{x\in C_1} \inf_{y\in C_2} \xi(x,y) = - \sup_{x\in C_1} \sup_{y\in
      C_2} (-\xi(x,y)) \in \mathbf{L}^1, 
  \end{displaymath}
  which, jointly with~\eqref{eq:98}, implies~\eqref{eq:89}. To
  verify~\eqref{eq:98} observe that the random field
  \begin{displaymath}
    \eta(y) \set \sup_{x\in C_1}\xi(x,y), \quad y\in V,
  \end{displaymath}
  has sample paths in the space of convex functions and, by
  Lemma~\ref{lem:24}, $\eta(y) \in \mathbf{L}^1$. Another application
  of Lemma~\ref{lem:24} yields $\norm{\eta}_{0,C_2}\in \mathbf{L}^1$,
  which clearly implies~\eqref{eq:98}.

  To verify \eqref{eq:90}, choose $\epsilon>0$ so that the sets
  $C_1(\epsilon)$ and $C_2(\epsilon)$ defined by \eqref{eq:91} are
  still in $U$ and $V$. Then, by Lemma~\ref{lem:23}, there is
  $c=c(\epsilon)>0$ such that for every $x\in C_1$ and $y\in C_2$
  \begin{align*}
    \norm{\xi(x,\cdot)}_{1,C_2} + \norm{\xi(\cdot,y)}_{1,C_1} & \leq
    c(\norm{\xi(x,\cdot)}_{0,C_2(\epsilon)} +
    \norm{\xi(\cdot,y)}_{0,C_1(\epsilon)}) \\
    &\leq 2c \norm{\xi}_{0,C_1(\epsilon)\times C_2(\epsilon)},
  \end{align*}
  and the result follows.
\end{proof}

\section{Stochastic fields of martingales}
\label{sec:modif-rand-fields}

This appendix contains the results concerning the sample paths of
stochastic fields of martingales used in the proofs of
Theorems~\ref{th:9} and~\ref{th:10}.

We fix a complete filtered probability space
$(\Omega,\mathcal{F},(\mathcal{F}_t)_{t\in [0,T]}, \mathbb{P})$ with
$\mathcal{F} = \mathcal{F}_T$. Recall the Fr\'echet spaces
$\mathbf{C}^m$ with semi-norms $\norm{\cdot}_{m,C}$ from
Section~\ref{sec:stab-under-conv}. As usual, $\mathbf{L}^0 =
\mathbf{L}^0(\Omega,\mathcal{F},\mathbb{P})$ denotes the space of
(equivalent classes of) random variables with convergence in
probability.

\begin{Lemma}
  \label{lem:25}
  Let $m$ be a non-negative integer, $U$ be an open set in
  $\mathbf{R}^d$, and $\map{\xi}{U}{\mathbf{L}^0}$ be a random field
  with sample paths in $\mathbf{C}^m = \mathbf{C}^m(U)$ such that for
  every compact set $C\subset U$
  \begin{displaymath}
    \mathbb{E}[\norm{\xi}_{m,C}] < \infty.
  \end{displaymath}
  Then the stochastic field
  \begin{displaymath}
    M_t(x) \set \mathbb{E}[\xi(x)|\mathcal{F}_t], \quad t\in [0,T], 
    \, x \in U, 
  \end{displaymath}
  has a modification with sample paths in
  $\mathbf{D}([0,T],\mathbf{C}^m)$ and, for every multi-index $k =
  (k_1,\dots,k_d)$ of non-negative integers with $\abs{k} \set
  \sum_{i=1}^d k_i \leq m$,
  \begin{displaymath}
    D^k M_t(x) = \mathbb{E}[D^k \xi(x)|\mathcal{F}_t],
    \quad t\in [0,T], \, x \in U, 
  \end{displaymath}
  where the differential operator $D^k$ is defined in \eqref{eq:56}.
\end{Lemma}
\begin{proof} By induction, it is sufficient to consider the cases
  $m=0,1$.  Assume first that $m=0$. It is well-known that, for every
  $x\in U$, the martingale $M(x)$ has a modification in
  $\mathbf{D}([0,T],\mathbf{R})$. Fix a compact set $C\subset U$ and
  let $(x_i)_{i\geq 1}$ be a dense countable subset of $C$. Standard
  arguments show that the stochastic field $\map{M}{C\times
    [0,T]}{\mathbf{R}}$ has a modification in
  $\mathbf{D}([0,T],\mathbf{C})$ if
  \begin{align}
    \label{eq:99}
    \lim_{a\to \infty} \mathbb{P}[\sup_{x_i} (M(x_i))^*_T \geq a] &= 0, \\
    \label{eq:100}
    \lim_{\delta \to 0} \mathbb{P}[\sup_{\abs{x_i-x_j}\leq \delta}
    (M(x_i) - M(x_j))^*_T \geq \epsilon ] &= 0, \mtext{for every}
    \epsilon >0,
  \end{align}
  where $X^*_t \set \sup_{0\leq s\leq t}\abs{X_s}$.

  From the conditions on $\xi=\xi(x)$ we deduce that the martingales:
  \begin{align*}
    X_t &\set \mathbb{E}[\sup_{x\in C} |\xi(x)| | \mathcal{F}_t], \\
    Y_t(\delta) &\set \mathbb{E}[\sup_{\abs{x_i-x_j}\leq \delta}
    |\xi(x_i) - \xi(x_j)| | \mathcal{F}_t],
  \end{align*}
  are well-defined and
  \begin{equation}
    \label{eq:101}
    \lim_{\delta\to 0} \mathbb{E}[Y_T(\delta)] = 0. 
  \end{equation}

  Since, clearly,
  \begin{align*}
    \sup_{x_i} |M_t(x_i)| &\leq  X_t, \\
    \sup_{\abs{x_i-x_j}\leq \delta} |M_t(x_i) - M_t(x_j)| & \leq
    Y_t(\delta),
  \end{align*}
  we deduce from Doob's inequality:
  \begin{align*}
    \mathbb{P}[\sup_{x_i} (M(x_i))^*_T \geq a] & \leq
    \mathbb{P}[X^*_T \geq a] \leq \frac1a \mathbb{E}[X_T],\\
    \mathbb{P}[\sup_{\abs{x_i-x_j}\leq \delta} (M(x_i) - M(x_j))^*_T
    \geq \epsilon ] &\leq \mathbb{P}[ (Y(\delta))^*_T \geq \epsilon ]
    \leq \frac1\epsilon \mathbb{E}[Y_T(\delta)],
  \end{align*}
  which, jointly with \eqref{eq:101}, implies \eqref{eq:99} and
  \eqref{eq:100}. This concludes the proof for the case $m=0$.

  Assume now that $m=1$ and define the stochastic field
  \begin{displaymath}
    \map{D_t(x) \set \mathbb{E}[\nabla \xi(x)|\mathcal{F}_t]}{U\times
      [0,T]}{\mathbf{L}^0(\mathbf{R}^d)},  
  \end{displaymath}
  with values in $\mathbf{R}^d$, where $\nabla \set (\frac{\partial
  }{\partial x^1}, \dots, \frac{\partial }{\partial x^d})$ is the
  gradient operator.  From the case $m=0$ we obtain that the
  stochastic fields $M=M_t(x)$ and $D= D_t(x)$ have modifications in
  $\mathbf{D}([0,T],\mathbf{C})$, which we shall use. For $M=M_t(x)$
  to have a modification in $\mathbf{D}([0,T],\mathbf{C}^1)$ with the
  derivatives given by $D=D_t(x)$ we have to show that
  \begin{equation}
    \label{eq:102}
    \lim_{\delta \to 0} \mathbb{P}[\sup_{x\in C, \abs{x-y}\leq \delta}
    \frac1\delta (N(x,y))^*_T \geq \epsilon] = 0, \mtext{for every}
    \epsilon>0, 
  \end{equation}
  where
  \begin{displaymath}
    N(x,y) \set M(y) - M(x) - \ip{D(x)}{y-x}. 
  \end{displaymath}

  We follow the same path as in the proof of the previous case. Our
  assumptions on $\xi = \xi(x)$ imply that, for sufficiently small
  $\delta >0$, the martingale
  \begin{displaymath}
    Z_t(\delta)  \set \mathbb{E}[\sup_{x\in
      C, |x-y|\leq \delta} \frac1\delta |\xi(y) -
    \xi(x) - \ip{\nabla \xi(x)}{y-x}|\; | \mathcal{F}_t], 
  \end{displaymath}
  is well-defined and
  \begin{equation}
    \label{eq:103}
    \lim_{\delta \to 0} \mathbb{E}[Z_T(\delta)] = 0. 
  \end{equation}
  Since
  \begin{displaymath}
    \sup_{x\in C, \abs{x-y}\leq \delta} \frac1\delta N_t(x,y) \leq Z_t(\delta), 
  \end{displaymath}
  we have, by Doob's inequality,
  \begin{align*}
    \mathbb{P}[\sup_{x\in C, \abs{x-y}\leq \delta} \frac1\delta
    (N(x,y))^*_T \geq \epsilon] \leq \mathbb{P}[(Z(\delta))^*_T \geq
    \epsilon] \leq \frac1\epsilon \mathbb{E}[Z_T(\delta)],
  \end{align*}
  and \eqref{eq:102} follows from \eqref{eq:103}.
\end{proof}

\begin{Lemma}
  \label{lem:26}
  Let $U$ be an open set in $\mathbf{R}^m$ (and, in addition, be a
  convex set for conditions \ref{item:17} and \ref{item:21} and a cone
  for \ref{item:18}). Let furthermore $V$ be an open set in
  $\mathbf{R}^l$, and $\map{\xi=\xi(x,y)}{U\times V}{\mathbf{L}^0}$ be
  a random field with continuous sample paths such that for every
  compact set $C\subset U\times V$
  \begin{displaymath}
    \mathbb{E}[\sup_{(x,y)\in C}\abs{\xi(x,y)}] < \infty.
  \end{displaymath}
  Then the stochastic field
  \begin{displaymath}
    M_t(x,y) \set \mathbb{E}[\xi(x,y)|\mathcal{F}_t], \quad 0\leq t\leq T,
    \, x \in U, \; y\in V,
  \end{displaymath}
  has a modification with sample paths in
  $\mathbf{D}([0,T],{\mathbf{C}(U\times V)})$. Moreover, if the sample
  paths of $\xi$ belong to $\widetilde{\mathbf{C}}$, then there is a
  modification of $M$ with sample paths in
  $\mathbf{D}([0,T],\widetilde{\mathbf{C}})$, where
  $\widetilde{\mathbf{C}}=\widetilde{\mathbf{C}}(U\times V)$ is any
  one of the following subspaces of $\mathbf{C}=\mathbf{C}(U\times
  V)$:
  \begin{enumerate}[label=(C\arabic{*}), ref=(C\arabic{*})]
  \item \label{item:15} $\widetilde{\mathbf{C}}$ consists of all
    non-negative functions;
  \item \label{item:16} $\widetilde{\mathbf{C}}$ consists of all
    functions $f=f(x,y)$ which are non-decreasing with respect to $x$;
  \item \label{item:17} $\widetilde{\mathbf{C}}$ consists of all
    functions $f=f(x,y)$ which are convex with respect to $x$;
  \item \label{item:18} $\widetilde{\mathbf{C}}$ consists of all
    functions $f=f(x,y)$ which are positively homogeneous with respect
    to $x$:
    \begin{displaymath}
      f(cx,y) = cf(x,y), \quad c>0. 
    \end{displaymath}
  \item \label{item:19} $\widetilde{\mathbf{C}}$ consists of all
    strictly positive functions;
  \item \label{item:20} $\widetilde{\mathbf{C}}$ consists of all
    functions $f=f(x,y)$ which are strictly increasing with respect to
    $x$:
    \begin{displaymath}
      f(x_1,y)<f(x_2,y), \quad x_1\leq x_2, \; x_1\not=x_2; 
    \end{displaymath}
  \item \label{item:21} $\widetilde{\mathbf{C}}$ consists of all
    functions $f=f(x,y)$ which are strictly convex with respect to
    $x$:
    \begin{displaymath}
      \frac12(f(x_1,y)+f(x_2,y))>f(\frac12(x_1+x_2),y), \quad x_1\not=x_2. 
    \end{displaymath}
  \end{enumerate}
\end{Lemma}

\begin{proof} The existence of a modification for $M$ with sample
  paths in $\mathbf{D}([0,T],{\mathbf{C}})$ has been proved in
  Lemma~\ref{lem:25}. Hereafter we shall use this modification.

  The assertions of items \ref{item:15}--\ref{item:18} are
  straightforward, since for every $t\in [0,T]$ these conditions are
  obviously satisfied for the random field
  $\map{M_t}{U}{\mathbf{L}^0}$ and the sample paths of $M$ belong to
  $\mathbf{D}([0,T],\mathbf{C})$.

  To verify \ref{item:19} recall the well-known fact that if $N$ is a
  martingale on $[0,T]$ such that $N_T>0$, then $\inf_{t\in
    [0,T]}N_t>0$. For every compact set $C\subset U\times V$ we have
  by~\ref{item:19} that $\inf_{(x,y)\in C} \xi(x,y)>0$ and, hence,
  \begin{displaymath}
    \inf_{t\in [0,T]}\inf_{(x,y)\in C} M_t(x,y) \geq \inf_{t\in
      [0,T]}\mathbb{E}[\inf_{(x,y)\in C} \xi(x,y)|\mathcal{F}_t] >0, 
  \end{displaymath}
  implying \ref{item:19}. Observe that this argument clearly extends
  to the case, when $U$ is an $F_\sigma$-set, that is, a countable
  union of closed sets.

  The cases \ref{item:20} and \ref{item:21} follow from \ref{item:19}
  by re-parameterization. For example, to obtain \ref{item:20} define
  the set $\widetilde U\subset {\mathbf{R}}^{2m}$ and the random
  fields $\map{\eta}{\widetilde U\times V}{\mathbf{L}^0}$ and
  $\map{N}{\widetilde U\times V\times [0,T]}{\mathbf{L}^0}$ as
  \begin{align*}
    \widetilde U &\set \descr{(x_1,x_2)}{x_i\in U, \; x_1\leq x_2, \;
      x_1\not=x_2}, \\
    \eta(x_1,x_2,y) & \set \xi(x_2,y) - \xi(x_1,y), \\
    N_t(x_1,x_2,y) & \set M_t(x_2,y) - M_t(x_1,y).
  \end{align*}
  While the set $\widetilde U$ is not open, it is an
  $F_\sigma$-set. An application of \ref{item:19} to $\eta$ and $N$
  then yields \ref{item:20} for $\xi$ and $M$.
\end{proof}

\begin{Lemma}
  \label{lem:27}
  In addition to the assumptions of Lemma~\ref{lem:26} suppose that
  \begin{equation}
    \label{eq:104}
    \mathbb{E}[\sup_{(x,y)\in U\times D} \xi(x,y)] < \infty,
  \end{equation}
  for every compact set $D\subset V$. Then the assertions of
  Lemma~\ref{lem:26} also hold for the following subspaces:

  \begin{enumerate}[label=(C\arabic{*}), ref=(C\arabic{*})]
    \setcounter{enumi}{7}
  \item \label{item:22} $\widetilde{\mathbf{C}}$ consists of all
    non-negative functions $f=f(x,y)$ such that for every increasing
    sequence $(C_n)_{n\geq 1}$ of compact sets in $U$ with
    $\cup_{n\geq 1} C_n = U$ and for every compact set $D\subset V$
    \begin{displaymath}
      \lim_{n\to\infty}\sup_{x\in U/C_n}\sup_{y\in D} f(x,y) = 0;
    \end{displaymath}
  \item \label{item:23} $\widetilde{\mathbf{C}}$ consists of all
    functions $f=f(x,y)$ such that for every increasing sequence
    $(C_n)_{n\geq 1}$ of compact sets in $U$ with $\cup_{n\geq 1} C_n
    = U$ and for every compact set $D\subset V$
    \begin{displaymath}
      \lim_{n\to\infty}\sup_{x\in U/C_n}\sup_{y\in D}f(x,y) = -\infty.
    \end{displaymath}
  \end{enumerate}
\end{Lemma}
\begin{proof}
  For the proof of \ref{item:22} recall that by Doob's inequality, if
  $(N^n)_{n\geq 1}$ is a sequence of martingales such that $N^n_T\to
  0$ in $\mathbf{L}^1$, then $(N^n)^*_T \set \sup_{0\leq t\leq
    T}\abs{N^n_t} \to 0$ in $\mathbf{L}^0$. Accounting for
  \eqref{eq:104} we deduce that, for the compact sets $(C_n)_{n\geq
    1}$ and $D$ as in \ref{item:22},
  \begin{displaymath}
    \lim_{n\to\infty}\mathbb{E}[\sup_{x\in U/C_n}\sup_{y\in D} \xi(x,y)]
    = 0. 
  \end{displaymath}
  The validity of \ref{item:22} for the sample paths of $M$ follows
  now from
  \begin{displaymath}
    \sup_{x\in U/C_n}\sup_{y\in D} (M(x,y))^*_T \leq \sup_{0\leq t \leq T}
    \mathbb{E}[\sup_{x\in U/C_n}\sup_{y\in D}\xi(x,y)|\mathcal{F}_t], 
  \end{displaymath}
  where we used the fact that in \ref{item:22} $\xi\geq 0$.

  Finally, \ref{item:23} follows from \ref{item:22} if we observe that
  a function $f=f(x,y)$ satisfies \ref{item:23} if and only if for
  every positive integer $n$ the function
  \begin{displaymath}
    g_n(x,y) \set \max(f(x,y)+n,0), \quad (x,y)\in U\times V,
  \end{displaymath}
  satisfies \ref{item:22}.
\end{proof}

\begin{Acknowledgments}
  We thank Andreas Hamel for references on the max-rule for
  subdifferentials used in Appendix~\ref{sec:envelope-theorem}.
\end{Acknowledgments}

\bibliographystyle{plainnat} \bibliography{../bib/finance}

\end{document}